%% file: ms.tex
\documentclass[letterpaper,twocolumn,10pt]{article}

\usepackage[nocompress]{cite}
\usepackage{amsthm}
\usepackage[usenames,dvipsnames,svgnames,table]{xcolor}
\usepackage{soul}
\usepackage[mathscr]{euscript}
\newcommand{\Note}[1]{}
\renewcommand{\Note}[1]{\hl{[#1]}}  
\newcommand{\NoteSigned}[3]{{\sethlcolor{#2}\Note{#1: #3}}}
\newcommand{\NoteME}[1]{\NoteSigned{ME}{YellowGreen}{#1}}  
\newcommand{\NoteSG}[1]{\NoteSigned{SG}{LightBlue}{#1}}    
\newcommand{\NoteMT}[1]{\NoteSigned{MT}{Pink}{#1}} 

\newtheorem{theorem}{Theorem}
\newtheorem{lemma}{Lemma}
\newtheorem{proposition}{Proposition}

\newtheorem{definition}{Definition}

\newtheorem{remark}{Remark}

\usepackage{multirow, amsmath}
\makeatletter
\newif\if@restonecol
\makeatother

\usepackage[titlenotnumbered,noend,noline]{algorithm2e}

\usepackage{usenix,epsfig}

\usepackage{comment}

\usepackage{graphicx}
\usepackage{caption}
\usepackage{subcaption}
\graphicspath{ {figs/} }
\usepackage{epstopdf}
\usepackage{url}
\usepackage[finalnew,ignoremode]{trackchanges}
\addeditor{mohamed}

\sloppypar
\begin{document}





%

\date{}

\title{Logic Locking for Secure Outsourced Chip Fabrication: A New Attack and Provably Secure Defense Mechanism}

%
%
%
%
%



\author{
{\rm Mohamed El Massad}\\
New York University
\and
{\rm Jun Zhang}\\
New York University
\and
{\rm Siddharth Garg}\\
New York University
\and
{\rm Mahesh V. Tripunitara}\\
University of Waterloo
} 

\maketitle

\begin{abstract}
Chip designers outsource chip fabrication to 
external foundries, but at the risk of 
IP theft. Logic locking, a 
promising 
solution to mitigate this threat, adds 
extra logic gates (key gates) and inputs 
(key bits) to the chip 
so that it functions correctly only when 
the correct key, known only to the designer but not 
the foundry, is applied. 
In this paper, we identify a new vulnerability in 
all existing logic locking schemes. 
Prior attacks on logic 
locking have assumed that, in addition to the 
design of the 
locked chip, the attacker has access to a 
working copy of the chip. Our attack does not 
require a working copy and yet we 
successfully recover a significant 
fraction of key bits from the design of the 
locked chip only. Empirically, we demonstrate the success of our attack on eight large benchmark circuits from a benchmark suite that has been tailored specifically for logic synthesis research, for two different logic locking schemes.
Then, to address this vulnerability, we initiate the 
study of provably 
secure logic locking mechanisms.
We formalize, for the first time to our knowledge, a precise notion of security for logic locking. We establish that any locking procedure that is secure under our definition is guaranteed to counter our desynthesis attack, and all other such known attacks. We then devise a new logic locking procedure,  Meerkat, that guarantees that the locked chip reveals no  information about the key or the designer's intended functionality. A main insight behind Meerkat is that canonical representations of boolean functionality via Reduced Ordered Binary Decision Diagrams (ROBDDs) can be leveraged effectively to provide security. We analyze Meerkat with regards to its security properties and the overhead it incurs. As such, our work is a contribution to both the foundations and practice of securing digital ICs.
\end{abstract}

\input{introduction.tex}

\input{background.tex}
\input{attack.tex}

\input{defense.tex}

\input{results.tex}

\input{discussion.tex}
\input{related.tex}

\input{conclusion.tex}


{\footnotesize \bibliographystyle{acm}
\bibliography{main}}


\end{document}

%% file: introduction.tex
\renewcommand{\paragraph}[1]{\vspace*{1em}\noindent\textbf{#1}\hspace*{1em}}
\section{Introduction}
\label{sec:intro}

The cost of setting up a semiconductor foundry 
has been increasing with technology scaling, and is 
currently upwards of \$5 billion~\cite{FoundryCost} for 
cutting-edge fabrication.
As a result, 
many semiconductor design companies have 
adopted the {fabless} model (foundries are also referred to as a fabs), 
i.e., they outsource integrated circuit (IC) 
fabrication to one of a few large commercial semiconductor 
foundries, often located off-shore. 
As per a recent study~\cite{website:garner}, four out of the 
top five semiconductor foundries by 
volume are located outside the United States. Most other countries  
either do not have a commercial 
foundry on-shore, and others that do only have access to low-end 
manufacturing technology.  

The fabless model comes with the risk of compromising the 
designer's intellectual property (IP), i.e., the chip's design details. 
When a designer outsources a chip for manufacturing, the 
foundry obtains full access 
to the chip's layout (effectively a blue-print of the chip), 
from which it can 
recover its \emph{netlist} (a network of interconnected 
Boolean logic gates) 
and, as a result, its 
Boolean 
functionality~\cite{torrance2009state}.
The 
foundry can then sell 
the designer's IP 
to its competitors. If the chip implements proprietary 
protocols or algorithms 
the designer stands to lose her competitive 
advantage. 
In addition, the foundry can also manufacture and sell extra copies of the 
chip in the black market. For these reasons, the risk of 
IP theft has become a serious concern both for commercial IC design
companies, as well as for
state actors like national defense agencies. 




\emph{Logic locking}, 
a technique first introduced by Roy et al.~\cite{roy2008epic}, 
is a promising solution to protect the designer's IP from 
theft by an untrusted foundry\footnote{Logic locking has 
also been referred to as logic obfuscation~\cite{rajendran2012security} and logic encryption~\cite{dupuis2014novel,rajendran2012logic,rajendran2015fault} in literature. Following the lead of Roy et al., we will use the term logic locking in our work.}.
Logic locking works by inserting 
additional gates, referred to as \emph{key gates}, 
in a netlist with 
side-inputs that are referred to as \emph{key bits}. 
The key bits 
are 
stored in a key register. 
The netlist functions as intended \emph{only} 
for a certain key (which we call the correct key), and provides incorrect
outputs otherwise. 
The correct key 
is known to the designer but not to the foundry. The 
logic locking procedure proposed by Roy et al. 
inserts XOR/XNOR gates at randomly chosen locations in the 
netlist, as shown in Figure~\ref{fig:llock}. Several other logic 
locking techniques 
have been subsequently proposed
~\cite{rajendran2012security,rajendran2012logic,rajendran2015fault,plaza2015solving,dupuis2014novel,chakraborty2009harpoon,baumgartenpreventing}, and are all premised on the same
basic idea. 

\begin{figure*}
  \centering\includegraphics[width=1\linewidth]{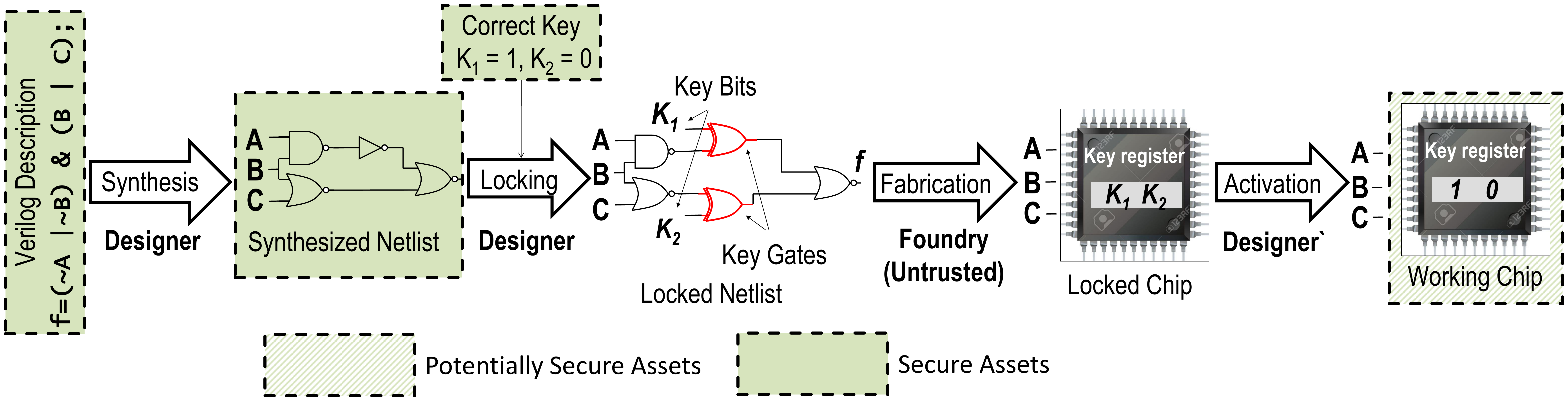}
  \caption{The XOR-based logic locking scheme proposed by Roy et al.~\cite{roy2008epic}. Assets in green boxes are assumed to be 
  secure, i.e., known 
  only to the defender. The working chip can be either secure or 
  available to the attacker depending on the threat scenario.} 
  \label{fig:llock}
\end{figure*}
Once a netlist is locked, it is sent to the foundry for fabrication.  
The foundry manufactures the locked netlist 
and ships manufactured chips to the 
designer. The designer activates chips by loading the 
correct key into each chip's key 
register\footnote{We note that 
the protocol proposed by Roy et al. 
also 
protects against a related threat, that of 
unauthorized over-building of the chip, 
using an additional key that is unique for
every chip.  
Section~\ref{sec:discussion} discusses the role of the unique key in more detail.}.
To ensure that the correct key cannot 
be read out by destructive chip imaging,
the key register is assumed to be implemented 
using read-proof memory~\cite{comerford1992tamper}.
Figure~\ref{fig:llock} depicts Roy et al.'s 
logic locking scheme, which they refer to as EPIC, 
with an example
(Section~\ref{sec:background} has more details on EPIC).

\paragraph{Attacks on logic locking} 
The security of logic locking 
is premised on the foundry not knowing the correct key.
However, El Massad et al.~\cite{el2015integrated}
and Subramanyan et al.~\cite{subramanyan2015evaluating} have
demonstrated that
the correct key can be compromised 
\emph{if} the foundry obtains a 
previously manufactured version of the chip from the market that 
has already been activated by the designer with the correct key.
In this attack scenario, the foundry has two assets: (1) a locked netlist; and 
(2) a working chip, i.e., one activated with the correct key.  
Note that the foundry cannot 
directly probe the working chip for the correct key --- recall the 
key is stored 
in read-proof memory --- 
but can observe the
chip's input/output (I/O) functionality. \cite{el2015integrated} and \cite{subramanyan2015evaluating}
leverage the chip's I/O functionality 
to reverse engineer the correct key 
using SAT-based inference algorithms. 
Empirically, these attacks demonstrate that 
only a relatively 
small number of I/O observations are required to recover 
the correct key, even for netlists locked with up to 128 key bits. 
The proposed attacks are effective against a wide range of logic 
locking techniques proposed in literature, including those proposed in \cite{roy2008epic}, \cite{baumgartenpreventing}, \cite{rajendran2012security}, \cite{rajendran2015fault} and \cite{dupuis2014novel}.


\paragraph{Our work} 
The existing attacks on logic locking 
make a strong assumption 
on the attacker's assets. 
That is, the 
attacker
is assumed to 
have access to a working copy of the chip (one that is 
previously manufactured and activated by the designer). 
However, in several important use-cases of logic locking, 
it is 
impractical for the attacker to obtain 
such a working copy. 
For one, a foundry is unlikely to 
obtain working copies of 
chips manufactured for non-commercial  
purposes, for instance chips manufactured for 
government agencies. 
These cannot be purchased in the market.
Second, even for consumer electronics, 
a foundry 
cannot acquire working copies of a
chip that is being manufactured for the first time. 
In fact, the first manufacturing run is exactly 
when a designer  
might be most interested in protecting the chip's IP 
using logic locking. 

This leads us to our motivating question: 
\emph{is logic locking secure when the  
attacker cannot acquire a working copy of the chip?} 
The implicit assumption in literature has been 
that logic locking is indeed secure under this constrained (but 
more practical) attack scenario. 
In fact, Roy et al. explicitly assert this claim
in their paper that ``[a locked netlist alone] gives \emph{no} 
criterion to check for possible [keys]"~\cite{roy2008epic}. 
Derivative schemes based on Roy et al.'s mechanism, including the ones proposed in~\cite{rajendran2012logic} and~\cite{rajendran2015fault}, 
implicitly inherit Roy et al.'s claim of security.

\paragraph{Paper contributions}
Our first contribution is to demonstrate
that logic locking is less secure than 
previously thought. 
Specifically, we refute the assertion in prior work 
that the the locked netlist alone provides no information about the 
correct key by demonstrating
a new attack that reverse 
engineers a \emph{significant} fraction of key bits 
with access to the locked netlist only, and not a 
working chip copy, as assumed by prior attacks. 
Our attack builds on the 
observation that all existing logic locking mechanisms attempt to lock the chip   
\emph{after} synthesis only\footnote{Synthesis is the process that ``compiles" a behavioral chip description to a hardware implementation, i.e., an optimized netlist of Boolean logic gates.}. As it turns out,  
information about the intended Boolean functionality is already embedded 
in the synthesized netlist and can be 
\change{reverse engineered}{extracted} by our attack. We therefore 
refer to our attack 
as a \emph{desynthesis attack}. 

The vulnerability exposed by the 
desynthesis attack 
motivates our second 
contribution: in Section~\ref{sec:notion} 
we formulate the first 
\emph{formal notion of security for logic locking}.  
Our security notion 
captures the requirement
that the locked netlist not 
provide any information to the attacker about the 
correct key (and, correspondingly, the designer's 
intended functionality).

Finally, we propose \emph{Meerkat}, a new logic locking 
mechanism that is provably secure under our 
notion of security. 
Instead of treating security as an after-thought, 
i.e., locking a synthesized netlist, 
Meerkat 
is a joint synthesis and logic locking algorithm. 
Meerkat directly outputs a locked netlist from a behavioural specification of the chip's functionality.

The key idea behind Meerkat is to
lock a canonical, reduced-ordered binary decision diagram (ROBDD) representation 
of the desired Boolean functionality instead of locking a synthesized netlist. 
Meerkat then synthesizes 
the locked ROBDD  into a 
locked netlist. 
The canonical nature of an ROBDD representation
allows us to prove that netlists locked with  
Meerkat do not reveal any information about the correct key.
The overheads of Meerkat are moderate, especially in 
light of the fact that Meerkat is the only logic 
locking tool that provides formal security guarantees.

%% file: background.tex
\section{Logic Locking Background}
\label{sec:background}

To lay the groundwork for our attack and proposed defense mechanism, 
we begin by describing 
existing logic locking schemes using the EPIC protocol as an exemplar.
Although logic locking is described somewhat informally in prior work, 
we adopt a more formal treatment of the topic for rigor and clarity. 


\paragraph{Preliminaries}
A (combinational) Boolean function $f\!\! : X \rightarrow Y$ 
where $X = \{0,1\}^n$ and $Y = \{0,1\}^m$ has $n$ inputs and
$m$ outputs. The function $f$ can be implemented using a 
Boolean netlist (or simply, netlist), $C = \{V_C, E_C\}$. $C$ can
be seen as a directed acyclic graph in which each vertex 
corresponds to a Boolean logic gate, one of the $n$ inputs, or
one of the $m$ outputs.
We can think 
of $C$ as a function as well, i.e., $C\!\!: X \rightarrow Y$ such that 
for $x \in \{0,1\}^n$, $C(x) = f(x)$, where $C(x)$ is
the output of the netlist for input $x$.

\newcommand{\lock}{\ensuremath{\textit{lock}}}

For a given Boolean function $f\!\!: X \rightarrow Y$ and a randomly picked $r$-bit
key $k^{*} \in K = \{0,1\}^{r}$, 
logic locking outputs a netlist 
$C_{\lock}\!: X \times K \rightarrow Y$ 
such that the following two requirements are satisfied:
\begin{enumerate}
    \item $C_{\lock}$ must 
agree with $f$ on all inputs when 
$k^*$ (the correct key) is applied, i.e., 
$$C_{\lock}(.,k^*)=f(.).$$
    \item $C_{lock}$
must differ from $f$ 
when 
a key other than the correct key is applied, i.e., for any key $k\neq k^*$, 
$$C_{\lock}(.,k)\neq f(.).$$
That is, the netlist can 
only be unlocked with the correct key.
Note that, in theory, this requirement is satisfied 
even if $C_{lock}$
differs from $f$ 
on only
one input 
for every incorrect key. In practice, 
it might be desirable 
for $C_{lock}$ to be \emph{substantially}
different from $f$ for an incorrect key. 
Rajendran et al.~\cite{rajendran2012logic} 
quantify this property 
using a metric, referred to as
\emph{output corruptibility}, 
that measures the average fraction
of inputs for which $f$ 
differs from $C_{lock}$ for any incorrect key. 
We report Meerkat's output corruptibility 
in Section~\ref{sec:results}. 
\end{enumerate}
Note that neither 
of the two 
requirements above 
address the security 
of the correct key itself. 
Indeed, there is no prior work of which we are aware that
articulates a precise notion of 
security of the correct key
in this context (other than an abstract and unproven 
assertion that the
locked netlist does not reveal the correct key). 
This is one of our
main contributions (see Section~\ref{sec:defense}). 

\paragraph{The EPIC protocol}
We now describe the ``Ending Piracy of ICs" for logic locking (EPIC)
protocol \cite{roy2008epic}.
The first step in EPIC, as in any conventional IC design process, is 
logic synthesis.

A synthesis tool is a program
that takes as input a Boolean function $f$, 
and transforms it 
into a Boolean netlist, $C$, that implements $f$. There may exist
several netlists that implement the function. Commercial 
synthesis tools attempt to output a netlist that minimizes metrics such as
power, delay and measures of area such as the gate count of the netlist. For
 example, a synthesis tool that targets only 
the gate count metric 
ideally outputs a netlist $C$ that minimizes $|V_C|$. We model the 
synthesis tool as a function $\mathscr{S}$ where 
$\mathscr{S}(f) = C$.

EPIC's locking procedure, $\mathscr{L}_E$, operates on the  synthesized
 netlist $C$, and the randomly generated $r$-bit key $k^*$, to produce a locked 
netlist $C_{\lock}$. That is,
$C_{\lock} =  \mathscr{L}_E(C, k^{*}) = \mathscr{L}_E(\mathscr{S}(f), k^{*})$.
As we discuss in Section~\ref{sec:attack}, the manner in 
which the locking procedure 
$\mathscr{L}_{E}$ is composed with the synthesis procedure $\mathscr{S}$ 
introduces a vulnerability.

\newcommand{\tuple}[1]{\ensuremath{\left\langle#1\right\rangle}}
 
EPIC's locking procedure can be explained using the 
example in  Figure~\ref{fig:llock}.
The synthesized netlist in the figure is
locked with an $r=2$ bit key, $k^*=10$.
For each bit of the key, 
EPIC selects either a wire or an inverter in $C$
depending on the value of the key bit. 
If the key bit is $0$, 
EPIC selects a wire; otherwise EPIC selects an inverter.
The wire or inverter is then 
replaced with an XOR 
gate, with the 
corresponding bit of the key register as one of its inputs. 
Note that the EPIC protocol adds both 
XOR and XNOR gates, the latter by 
complementing the key value for which an inverter and wire
are selected.

In our example, as 
the first 
bit of the key is $1$, 
EPIC picks an inverter, i.e., 
the inverter that immediately follows the NAND gate, 
and replaces the 
inverter with an XOR. The 
unconnected input of the 
XOR gate is driven by the first key bit. 
As the second key bit is $0$, EPIC 
replaces the wire between the two NOR gates 
with an XOR driven by the second key bit. 

The functionality of the netlist remains 
unchanged when the correct key 
is loaded into the key register. In our example, 
if the correct key $k^*=10$ is applied, 
the first XOR gate acts like an inverter
while the second XOR gate acts like a wire. 
Setting any key bit incorrectly causes a wire in the original netlist 
to behave as an inverter, or vice-versa.
Ideally, this changes the functionality of the 
netlist, but 
the authors of EPIC point out some pathological 
cases in which even 
incorrect keys result in correct functionality 
(in effect, because the impact of two or more incorrect 
key bits cancels out).  

Once a netlist has been locked, it 
is used as the basis 
for all remaining steps in the design flow, such as gate placement and wire 
routing. These steps are not pertinent to the security of
logic locking and are therefore 
omitted from Figure~\ref{fig:llock}, as well as the remainder of
 our discussions. The final result is a chip 
layout file that is sent to the foundry for manufacturing. 
The foundry manufactures the chip and ships the manufactured parts to the 
designer, who can then activate the chips by loading the 
correct key into the key register. 

We note that like EPIC, all other   
logic locking schemes in literature~\cite{baumgartenpreventing,rajendran2012security,dupuis2014novel}
first synthesize 
$f$ to generate a netlist $C$ and then
use
different procedures to lock $C$.

\newcommand{\dummyfig}[1]{
  \centering
  \fbox{
    \begin{minipage}[c][0.20\textheight][c]{0.45\textwidth}
      \centering{#1}
    \end{minipage}
  }
}

%% file: attack.tex
\section{Desynthesis Attack}\label{sec:attack}
In this section, we describe our new attack, 
the \emph{desynthesis attack},
on logic locking. We begin by describing 
our attack scenario.
Unlike previous attacks, we do not assume access to a working copy of the chip (or, in general, to any I/O behavioral information). 
We then illustrate the vulnerability that our attack exploits using a simple example. Finally, we describe a concrete way that an attacker might make use of the vulnerability. 

\input{attack-scenario}

\input{vulnerability}

\subsection{Implementation of Desynthesis Attack}

We now describe the implementation of 
our desynthesis attack. 
The attack 
aims to find a $k$ that 
results in a   
re-synthesized netlist, $C_{k} = \mathscr{S}(f_{k})$,   
that is the most ``similar" to the 
locked netlist $C_{lock}$ (minus its key gates). 

Because an exhaustive search for $k^*$ is computationally 
infeasible, 
we instead adopt a greedy search heuristic. 
We start with a random guess for $k^*$,
and we iteratively improve upon the guess by exploring the immediate neighborhood of keys that are one 
Hamming distance away from the current guess. 
From this set of neighboring keys, we pick the one 
that results in a re-synthesized netlist 
most similar to $C_{lock}$. 
The iterations terminate when all 
re-synthesized netlists in the neighborhood of the current 
guess are \add{less} similar to  
$C_{lock}$ than the re-synthesized netlist corresponding to the 
current guess.
To better explore the key space, 
we perform the local search multiple times, each time starting with a different initial key guess.

The description so far has assumed a metric to quantify 
similarity (or, more precisely, dis-similarity) between netlists. 
Measures of similarity between graphs have been studied widely in 
literature~\cite{ged}. 
However, many of these measures, for instance graph 
edit-distance, are expensive to compute. 
We found, empirically, that the following simple measure of dis-similarity, $\Delta$, not only works well in practice but is also easy to compute. 
Specifically, 
$$ \Delta(C_1,C_2) = \sum_{i=1}^l (n_{g_i}[C_1]-n_{g_i}[C_2])^2 $$
where $C_1$ and $C_2$ are the two netlists to be compared, $\{g_1,\ldots,g_l\}$ is the set of gate types in the standard cell library\footnote{The library of Boolean 
logic gates the synthesis tool is allowed to pick from}, and $n_G[C]$ is the number of gates of type $G$ in netlist $C$. 

Algorithm \ref{attack_alg} describes our desynthesis attack. In Lines (1)--(2) of the algorithm, we count the number of key bits in $C_{lock}$, and set $\mathit{abs\_min}$, a variable that holds smallest dis-similarity value found so far, to a very large value, and in the ``foreach'' loop, we perform a local search for the key that minimizes the dis-similarity measure between the re-synthesized netlist and $C_{lock}$.



\vspace*{0.125in}
\noindent
\textbf{Inputs}: (1) A locked netlist $C_{\lock}$, and,
 (2) $\mathit{restarts}$, the number of restarts for the local search.

\vspace*{0.0625in}
\noindent
\textbf{Output}: A guess for the correct key, $k^*$, that was used to generate
$C_{\lock}$.
\LinesNumbered
\DontPrintSemicolon
\begin{algorithm}
$r \gets \text{number of key bits in } C_{\lock}$\;
$\mathit{abs\_min} \gets \infty$\;
\ForEach{$i$ from $1$ to $\mathit{restarts}$}{
    $k^0 \gets \text{random } r\text{-bit string}$\;
    \While{true}{
	$k^1 \gets \arg\min_{k: \mathit{ham}(k,k^0) \leq 1} \Delta \Big( C_{\lock}, \mathscr{S}\big(f_k\big)\Big)$\;
	\lIf{$k^1 = k^0$}{ break }
	\lElse{$k^0 \gets k_1$}
    }
    $\mathit{min} \gets \Delta \Big( C_{\lock}, \mathscr{S}\big(f_{k^1}\big)\Big)$\;
    \If{$\mathit{min} < \mathit{abs\_min}$}{
        $\mathit{best\_guess} \gets k^1$\;
        $\mathit{abs\_min} \gets \mathit{min}$\;
    }
}
\Return{ $\mathit{best\_guess}$ }\;

\vspace*{0.125in}
\caption{Desynthesis Attack}

\label{attack_alg}
\end{algorithm}



 \begin{figure*}[!ht]
\centering
    \begin{subfigure}[b]{0.4\textwidth}
	\includegraphics[width=\textwidth]{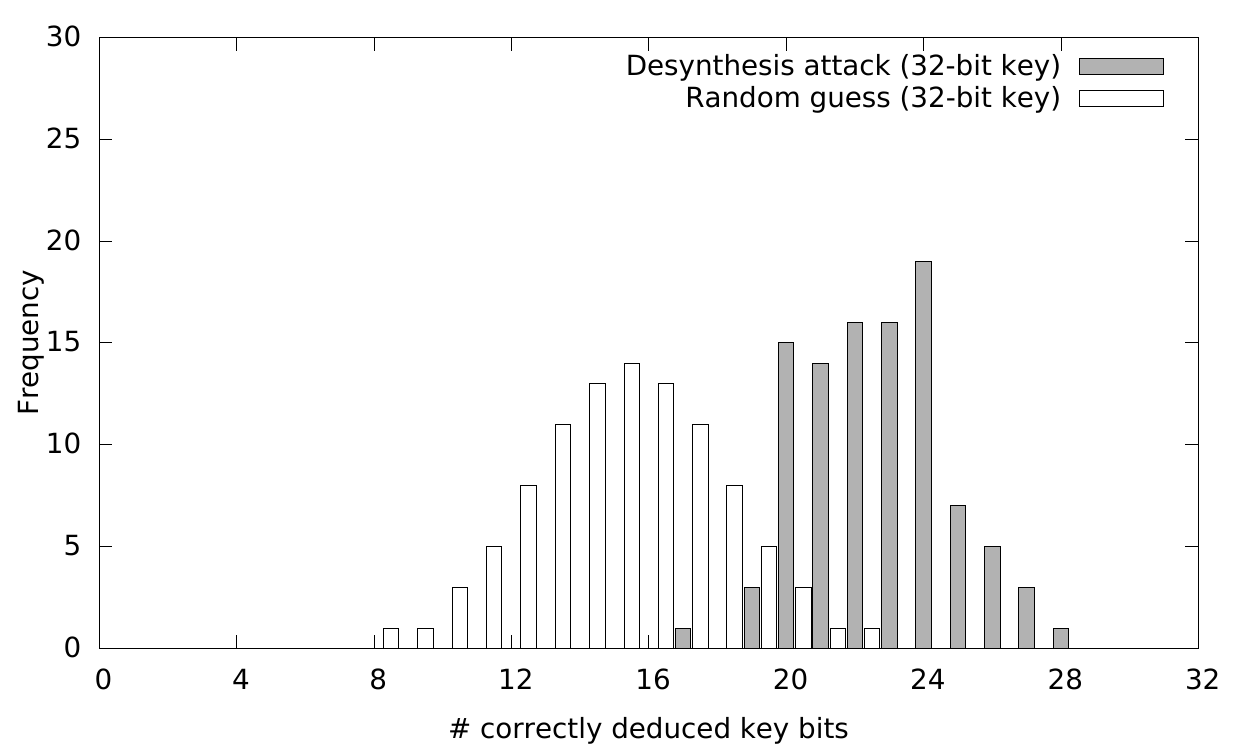}
        \caption{Results for 32-bit key}
        \label{fig:gull}
    \end{subfigure}
    \hspace{10 mm} 
    \begin{subfigure}[b]{0.4\textwidth}
	\includegraphics[width=\textwidth]{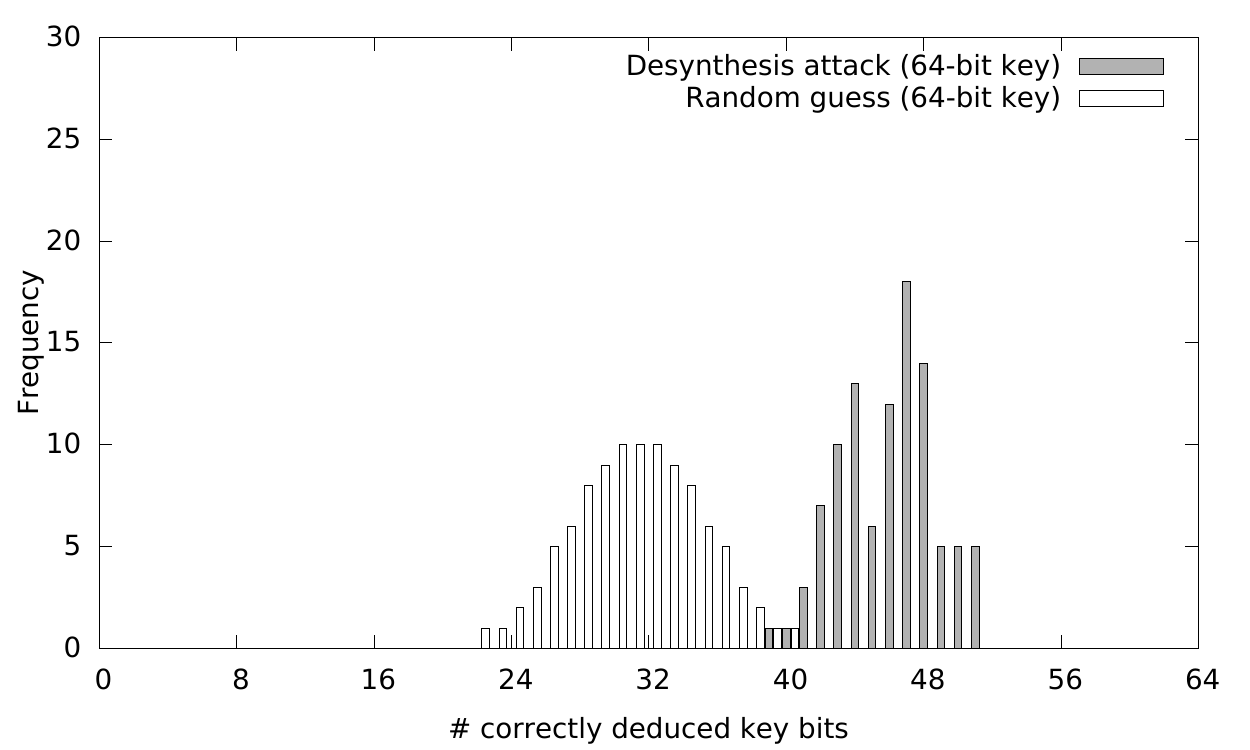}
        \caption{Results for 64-bit key}
        \label{fig:tiger}
    \end{subfigure}
  \caption{Attack results for styr benchmark from the MCNC benchmark suite.}
  \label{fig:scf-epic}
\end{figure*}

\paragraph{Effectiveness of desynthesis attack}
While we defer a comprehensive analysis of 
the results of our desynthesis attack to Section~\ref{sec:results}, 
we plot in 
Figure~\ref{fig:scf-epic}
the results of our attack on the styr benchmark from the 
MCNC benchmark set~\cite{LGSynth89}.
The benchmark is synthesized using the ABC synthesis tool, 
and the netlist is locked using EPIC, with 32- and 64-bit keys. 
Each histogram shows the number of bits of the correct key that we correctly recover over 100 experiments (we select the correct key and the locations of key gates independently at random in each experiment). 
Also shown, for comparison, are histograms of the number of 
key bits that 
an attacker would correctly guess 
if indeed the EPIC locking 
procedure were as secure as claimed (recall that the claim is that 
the locked 
netlist provides ``no criteria" to guess the key bits).
Note that in \emph{every} experiment, 
the desynthesis attack 
recovered more than half the key bits correctly. 
On average, the 
desynthesis attack recovers 23 (and up to 29) and 47 (and up to 59) 
key bits correctly for 32- and 64-bit keys, 
respectively. 

Based on these results, it is clear that 
the locked netlist 
\emph{does} 
indeed 
leak information about the designer's correct key, pointing to the 
insecurity of EPIC. In Section~\ref{sec:results} we perform 
statistical hypothesis tests 
of this assertion
against the null hypothesis, i.e., that the key and locked netlist 
are independent. We also report results on a larger set 
of benchmark circuits, two different synthesis tools and 
two different locking procedures.

%% file: attack-scenario.tex
\subsection{Attack Scenario}

The attacker in our scenario is a foundry that manufactures a 
locked netlist, $C_{\lock}$.
We assume that the foundry, as an attacker, has access to:
(a) the locked netlist, $C_{\lock}$,
(b) knowledge of which inputs in the netlist are regular inputs
and which are key bits,
(c) knowledge of the synthesis tool, $\mathscr{S}$, that is used by
the designer, and,
(d) the locking procedure that is used to generate
$C_{\lock}$ from $C$.

With regards to (a), the foundry obtains $C_{\lock}$ by reverse-engineering it from its layout, using a process
 called extraction. With regards to (b), the foundry can learn this separately by exploiting implementation artifacts, for example, the key bits are either read from read-proof memory or, in the full EPIC protocol, are outputs of a decryption module. Regular inputs, on the other hand, connect to chip I/O. 
With regards to (c) and (d), we note that 
Kerchoffs's principle \cite{petitcolas2011kerckhoffs} is relevant in this context. It states
 that, ``a (crypto)system should be secure even if everything about the system, except the key, is public knowledge.'' 
On (c), however,  
the designer may have used
a closed-source commercial logic synthesis tool. 
An option then is to weaken the attacker by assuming
that she has only black-box access to $\mathscr{S}$ by purchasing it
from the market. Indeed, our attack (see Section \ref{sec:attack}) 
only requires
such a weaker attacker to be practical.

However, when considering a defense mechanism 
(see Section \ref{sec:defense}),
we adopt the stronger
attacker that has open access to our locking 
and synthesis procedures. 

 

When compared to prior work, we assume that the attacker does not
have access to an activated and working copy of the chip.
We discuss our rationale for this under ``our work'' in Section \ref{sec:intro}.
We assume also that the attacker (the foundry) does not have any prior
information about $f$. That is, from the 
foundry's perspective, 
all different $2^{r}$ functions that 
$C_{lock}$ implements are equally likely.
%
 
\paragraph{Attacker's goal.}
Under the attack scenario above, the attacker's goal is to compromise the 
correct key $k^*$. Compromise, in this context, means that she discovers
several of the key bits.
An explicit assertion in prior work (see 
the quote from Roy et al.'s paper in Section~\ref{sec:intro}) 
is that with access to $C_{\lock}$ only,
the attacker can do no better than randomly guessing $k^{*}$. 
We show that, unfortunately, existing 
locking mechanisms fail 
to provide this property. Meerkat provably addresses 
this vulnerability.

%% file: vulnerability.tex
\subsection{The Vulnerability}
We now illustrate the vulnerability 
that our attack exploits 
using the logic locking 
example design from Figure~\ref{fig:llock}. 
The function that the designer
wishes to implement is 
$f=(\bar{A}+\bar{B}).(B+C)$.
The synthesized netlist $C$ 
corresponds to the lowest area 
implementation of $f$ using a
technology library  that comprises
two-input NAND, NOR and NOT (inverter) 
gates. (In this example,  
we assume that 
the synthesis procedure $\mathscr{S}$ 
aims to minimize area, but this is not necessary 
for our attack to succeed.)
As described in Section~\ref{sec:background}, 
$C$ is then 
locked 
with $k^*=\{1,0\}$. 


Now suppose the attacker wants to check 
if some key $k \in K$ is the correct key. 
The attacker can determine 
the Boolean function 
$f_{k} = C_{lock}(.,k)$ 
that the locked netlist implements with key $k$.
For example, say $k = 00$. As shown in 
Figure~\ref{fig:desynth}, the Boolean function 
corresponding to this key is $f_{00} = A.B$.
The attacker can now ask: could the  
locking procedure applied on $f_{00} = A.B$ 
produce $C_{lock}$?

\begin{figure}
\centering
\includegraphics[width=0.9\columnwidth]{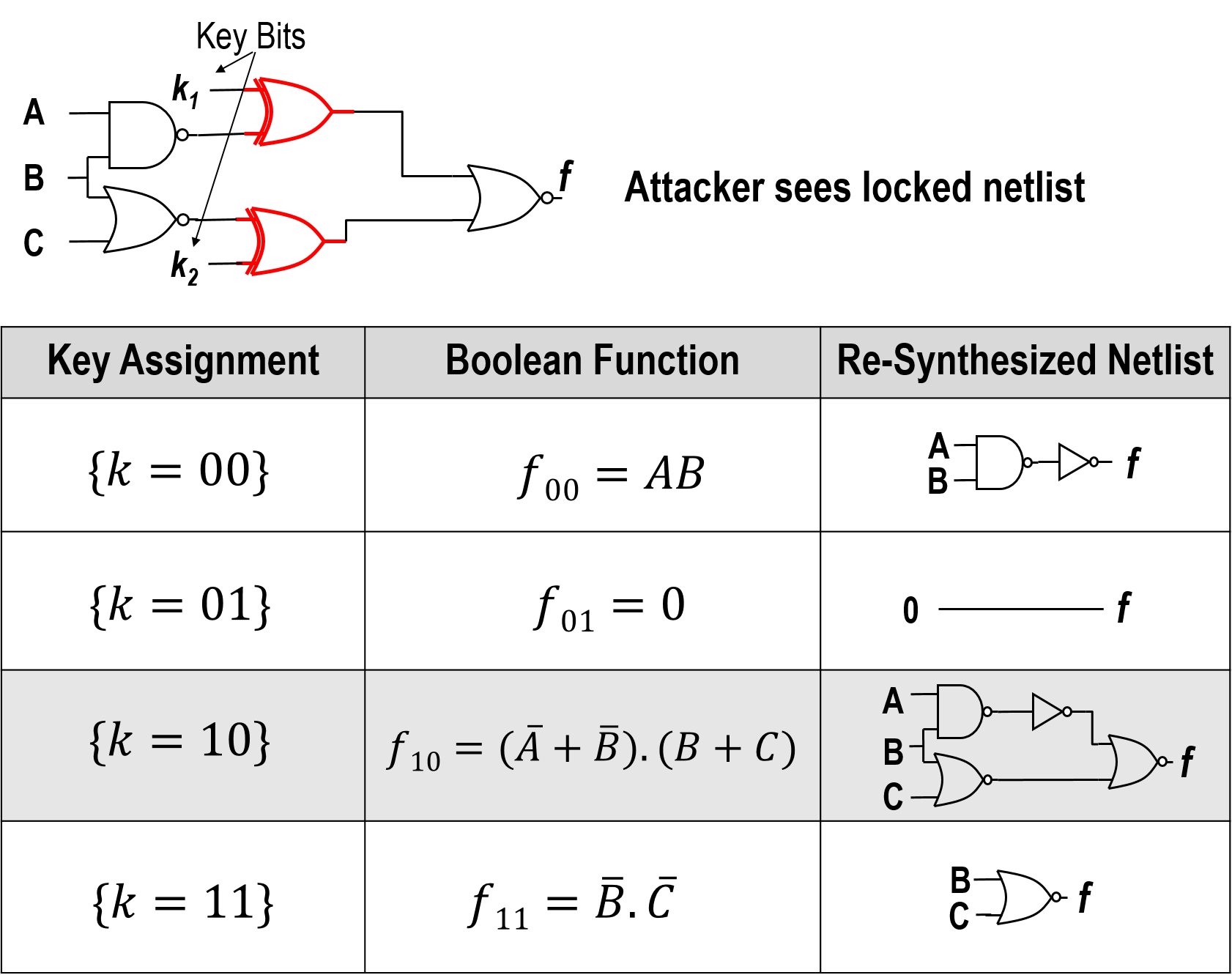}
\caption{An illustration of the desynthesis attack on the locked netlist from Figure~\ref{fig:llock}. 
}
\label{fig:desynth}
\end{figure}

As the first step in EPIC is synthesis, 
the attacker synthesizes $f_{00}=A.B$. (Recall, from the previous
section, that 
the attacker has at least 
black-box access to the defender's 
synthesis procedure). Figure~\ref{fig:desynth} 
shows the resulting synthesized netlist, 
$C_{00} = \mathscr{S}(f_{00})$.

 This netlist only has 
two gates --- a NAND gate and an inverter. 
Locking this netlist can introduce at most two
additional gates; the new gates must be 
either XOR or XNOR gates. On the other 
hand, $C_{lock}$ has \emph{two} NOR gates.     
There is no way to lock $C_{00}$ using the EPIC 
locking mechanism that would result in $C_{lock}$.   
The attacker can therefore eliminate $k = 00$ as the correct key. 
Using similar reasoning, the attacker eliminates  
$k = 11$ and $k = 01$ as correct keys. 

In the example, the Boolean function that corresponds to $k = 01$ is 
$f_{01} = 0$, a wire. And an attacker may use this as a basis for
eliminating $k = 01$ as a possibility, as a wire is unlikely to have been
the designer's intent.
This is certainly plausible. In our attack that we discuss in the next section,
we do not require such \emph{a priori} knowledge.

\paragraph{Practical challenges} 
The preceding example illustrates the 
strength of desynthesis attacks, albeit in the idealized 
setting that the attacker is able to brute-force the key.
 Implementing this attack in a real-world setting 
introduces some challenges. (1) An attacker cannot 
exhaustively search over all keys for large key sizes.
And, (2) commercial synthesis tools are imperfect. 
For the same $f$, the output 
of the synthesis tool can depend on the manner in which $f$ 
is described. Furthermore, existing  
synthesis tools are heuristic in nature --- they do not guarantee
an optimal solution, for example, a netlist with the 
lowest area, delay or power.
We now discuss a practical desynthesis attack procedure that addresses 
these challenges.

%% file: defense.tex
\section{Provably Secure Logic Locking}
\label{sec:defense}

We now discuss Meerkat, 
a new \remove{secure} logic synthesis and locking \add{procedure} that is 
provably secure in our attack scenario. 
We begin by 
formalizing 
a notion of security for logic locking
and then describe Meerkat, a candidate 
logic  
locking procedure that is 
provably secure under this notion of security.

\input{notion}

\subsection{Meerkat}
We now describe Meerkat, 
the first logic locking scheme that 
is provably secure as per Definition~\ref{def:sec}, and therefore
guarantees
that the 
locked netlist does not reveal any information about the correct key 
(see Section~\ref{sec:attack} for more details on the attack scenario). 
Central to Meerkat are reduced ordered binary decision diagrams (ROBDDs) --- 
canonical representations of 
Boolean functionality first introduced in a landmark paper by 
Bryant~\cite{bryant1986graph}. 
ROBDDs 
have found wide applicability in several areas of
hardware design,
particularly in the area of formal verification~\cite{robdd1,robdd2}, but 
have now largely been superseded by SAT/SMT 
solvers in industry. 
Meerkat revives the use of ROBDDs in the hardware security context.
We begin our description of Meerkat with a short ROBDD primer.


\subsubsection{An ROBDD Primer}
An ROBDD is a special case of a 
binary decision diagram (BDD), a data structure that allows 
a Boolean function to be evaluated 
as essentially a branching program. 
All but the terminal nodes of a BDD 
correspond to Boolean variables, and each has two children --- a high-child and 
a low-child.  For a given assignment of variables, a function can 
be evaluated starting from the root, and branching to the high-child if the 
corresponding variable is $1$ and to the low-child otherwise.
This continues till a terminal node, either a $0$ or a $1$ is reached.
 
Bryant's key observation was that if 
(a) certain reduction 
rules are followed to ensure that 
there are no redundant nodes in a BDD, and (b) variables in any path from root to 
a terminal node only occur in a fixed order, then the resulting BDD, called an R(educed)O(rdered)BDD, is a \emph{canonical} 
representation of Boolean functionality. 
Importantly, Bryant observed that ROBDD representations for many Boolean 
functions of practical interest are \emph{compact}, unlike other canonical 
representations like truth tables
(although ROBDDs can be
exponentially sized in the worst case.)  
As an example, an 
ROBDD for representing the function $f = \bar{x_1}\bar{x_2}\bar{x_3}+x_1x_2+x_2x_3$ is shown in Figure \ref{fig:bdd}.

\begin{figure*}[t]
\centering
\subcaptionbox{ROBBDD of $f = \bar{x_1}\bar{x_2}\bar{x_3}+x_1x_2+x_2x_3$. High (low) children shown using solid (dashed) lines.
\label{fig:bdd}}
[0.3\textwidth]{\includegraphics[width=0.18\textwidth]{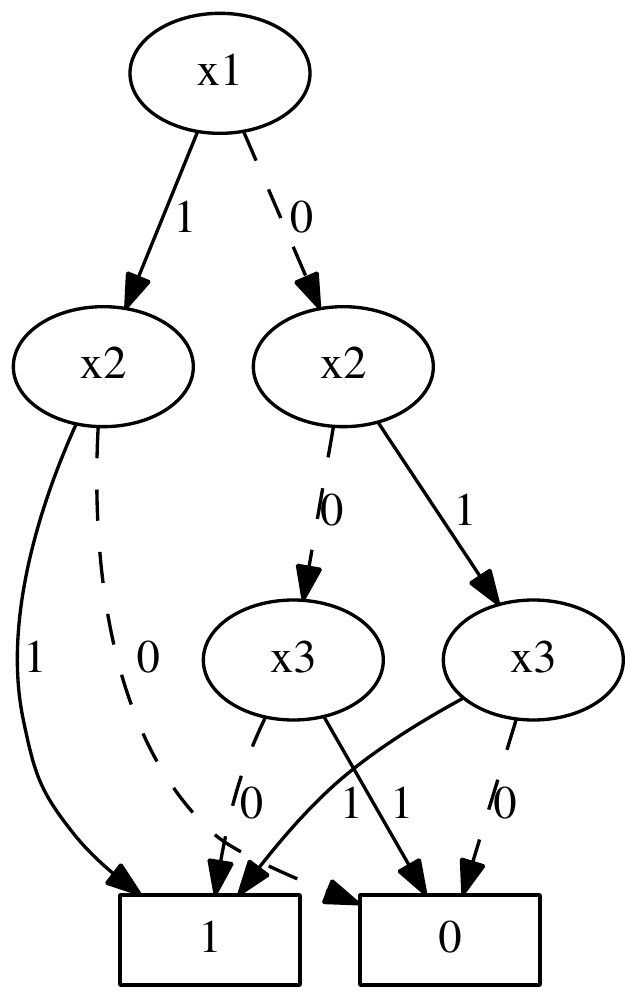}}~
\subcaptionbox{Locked ROBDD 
generated using correct key $k^*=10$. New key nodes are shown in grey.
\label{fig:bdd_lock}}
[0.35\textwidth]{\includegraphics[width=0.18\textwidth]{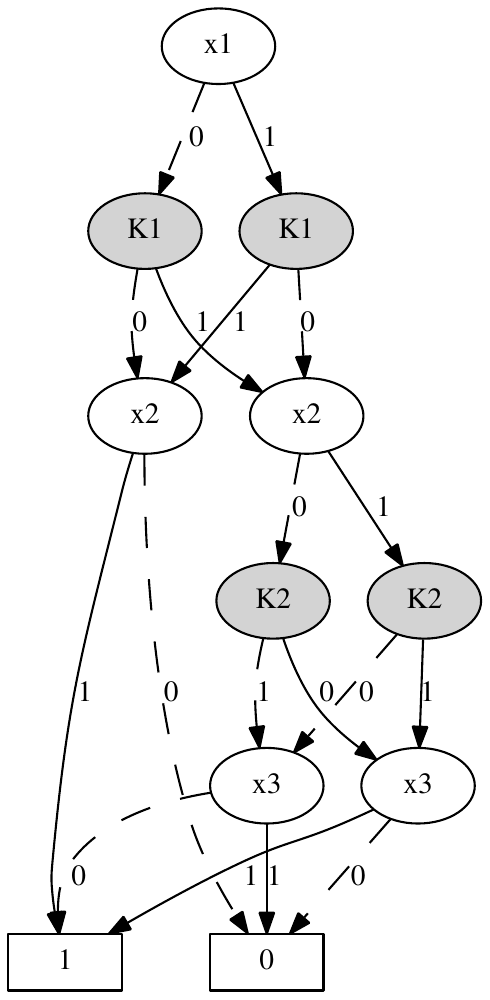}}~
\subcaptionbox{Netlist of MUXes that implements the locked ROBDD.\label{fig:mux}}
[0.35\textwidth]{\includegraphics[width=0.24\textwidth,angle=270]{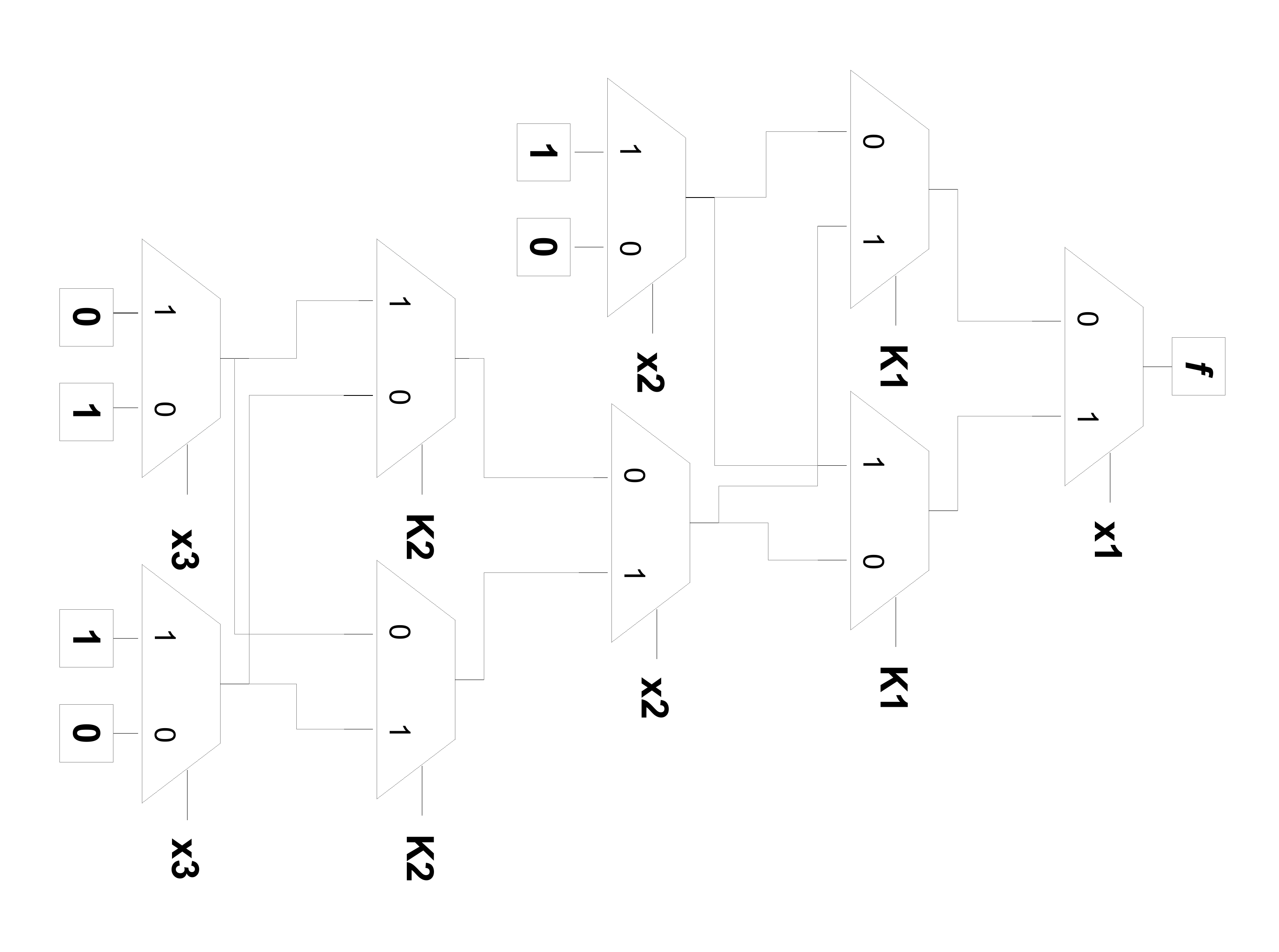}}
\caption{Meerkat's locking technique illustrated with an example. In the figures, the 
high and low children of a decision node are discriminated by using a solid (dashed) line to draw the edge leading to the high (low) child. }
\label{fig:meerkat}
\end{figure*}

More formally, an ROBDD for function $f: X \rightarrow \{0,1\}$ 
is a rooted directed acyclic graph, $G$, with
two sets of nodes: a set of \emph{decision} nodes, $D$, 
and two \emph{terminal} nodes, $0$ and $1$. 
The ROBDD has three associated functions. Functions 
$hi: D \rightarrow D \cup \{0\} \cup \{1\}$ and $lo: D \rightarrow D \cup \{0\} \cup \{1\}$ 
return the high-child  and low-child of each node, respectively. 
Function $var: D \rightarrow X$ returns the variable associated with 
each node.

Each node $d \in D$ implements a Boolean function $f_d$ which can be 
expressed as follows:
$$ f_d = var(d)f_{hi(d)} + \overline{var(d)} f_{lo(d)}.$$
The terminal nodes implement $f_1=1$ and $f_0=0$, and the root node 
implements $f$.

\subsubsection{Meerkat Procedure}

We now describe Meerkat's locking procedure. 
Starting with a behavioral description 
$f$, Meerkat first constructs its ROBDD, 
$G$. 
Then, Meerkat locks $G$ with key $k^*$
to output an new locked 
ROBDD, $G_{lock}$, 
as follows.

For each bit 
of the $r$-bit key, 
Meerkat picks a node of the ROBDD $G$ 
uniformly at random. 
Say for the $i^{th}$ key bit, node 
$d \in D$ 
is picked. We now add two new 
\emph{key} nodes to the $G$, $p_i$ and 
$q_i$ with the following properties:
\begin{itemize}
\item $var(p_i) = var(q_i) = {k^*}_{i}$, that is, the variable 
corresponding to both nodes is the $i^{th}$ key bit.      
\item If ${k^*}_i=1$ then 
$lo(p_i) = lo(d)$, $hi(p_i)=hi(d)$ and $lo(q_i) = hi(d)$, $hi(q_i)=lo(d)$. The high and low children are reversed if $k_i=0$.
\item $lo(d) = p_i$ and $hi(d) = q_i$, that is, the low-child 
and high-child of $d$ are redirected to $p_i$ and 
$q_i$, respectively. 
\end{itemize} 
There is one exception to this step. 
If there exists another node $d'$ in $G$ such that 
$var(d') = var(d)$, $hi(d')=lo(d)$ and $lo(d')=hi(d)$ 
then $d$ is discarded and a new node is selected. 
$d$ and $d'$ are referred to as 
\emph{complementary} nodes. 
We will see why this is required shortly. 
Figure~\ref{fig:bdd_lock} shows the result of locking 
the ROBDD in Figure~\ref{fig:bdd} with a 2-bit key, $k^*=10$. 

Finally, the 
locked ROBDD, $G_{lock}$ can be
synthesized into a netlist 
$C_{lock}$ using specialized BDD synthesis tools 
like BDS~\cite{bds} and FDD~\cite{fdd}.
Alternatively, any ROBDD (in fact any BDD) can be 
implemented using a network of multiplexers (MUXes) 
by replacing 
each node with a MUX, 
as shown in Figure~\ref{fig:mux}. 
This MUX-only netlist can then 
be fed into any 
standard synthesis tool 
(for further optimization) and 
outputs the final locked netlist. 
The astute reader might wonder if the use of 
a standard synthesis tool in the back-end compromises 
security. The answer is no: a striking feature of Meerkat 
is that it enables any  functionality-preserving 
transformation to be applied on 
$G_{lock}$ without compromising security at all (more details 
in Theorem~\ref{thm:main}).


\subsubsection{Security Analysis of Meerkat}

We now prove that Meerkat's locking procedure meets the 
security criterion prescribed in Definition~\ref{def:sec}. 
To begin, let
$G_{k}$ be the resulting BDD 
when key $k \in K$ is ``applied" to
$G_{lock}$, i.e., the key nodes in $G_{lock}$ are eliminated as follows. 
Let $d_i$ be the node that was randomly selected during the 
$i^{th}$ iteration of the locking procedure (corresponding to the $i^{th}$ bit of the key). 
For each $d_i$, 
\begin{itemize}
\item if $k_i = {k^*}_i$, 
then the high-child (resp. low-child) 
of $d_i$ in $G_{lock}$ is set to the high-child (resp. low-child) of $d_i$ in $G$, else
\item if $k_i \neq {k^*}_i$, 
then the high-child (resp. low-child) 
of $d_i$ in $G_{lock}$ is set to the low-child (resp. high-child) of $d_i$ in $G$.
\end{itemize} 
That is, the BDD $G_{k}$ can be obtained from ROBDD $G$ by 
flipping the high and low children of a subset of its vertices. 
(If the distinction between edges pointing to high versus low children is 
ignored, $G$ and $G_{k}$ are isomorphic.)

Note that the 
BDD $G_k$ represents the function 
$f_k = C_{lock}(.,k)$ that results when key $k$ is applied to the locked netlist.
We now show that: (1) $G_k$ is indeed an ROBDD and (2) $f_k$ is different from $f = f_{k^*}$ for 
$k \neq k^*$.

\begin{figure}
	\centering
		\includegraphics[width=\columnwidth]{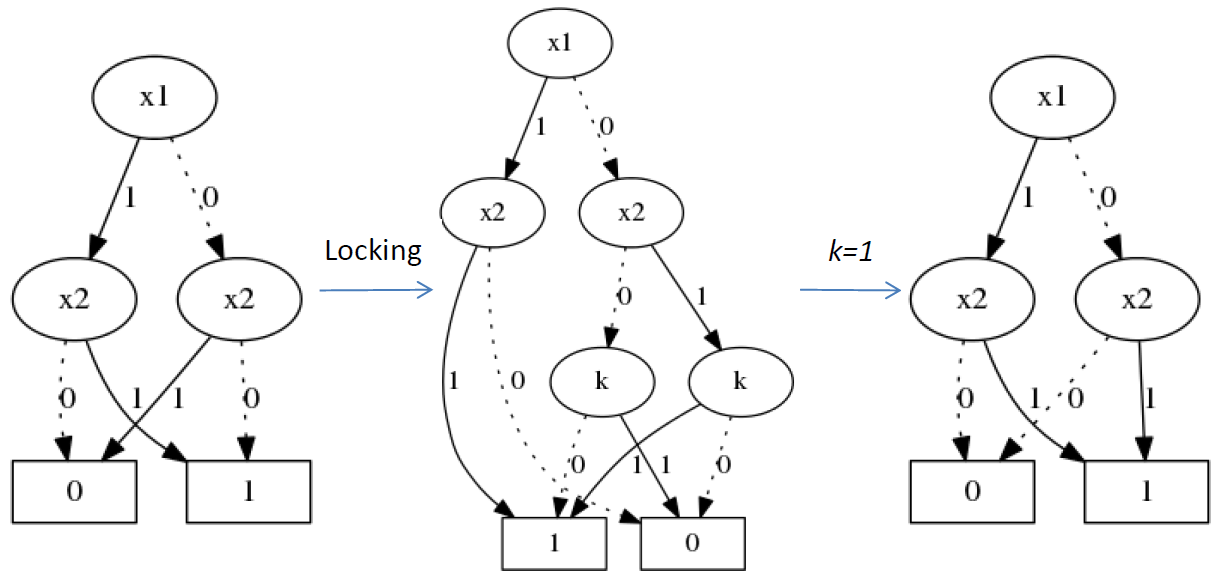}
	\caption{An example of an ROBDD with complementary nodes corresponding to the 
	variable $x_2$. One of the two nodes 
	is locked with $k=0$. Applying the locked ROBDD with $k=1$ results in BDD that is not 
	reduced.}\label{fig:comp_nodes}
\end{figure}

\input{table1}

\begin{lemma}
\label{lem:robdd}
$G_{k}$ is an ROBDD. 
\end{lemma}

\begin{proof}
We prove this assertion using induction on the 
number of key bits. For a single key bit, node 
$d_1$ in $G$
is selected and its high- and low-children flipped to yield 
$G_k$. 
The high- and low-children of $d_1$ cannot be isomorphic (or $G$ 
would not be be an ROBDD). Furthermore, as long as $d$ does not 
have a complementary node in $G$, 
no two nodes in $G_k$ are isomorphic (or again $G$ would not be an ROBDD). 
However, if $d$ has a complementary node in $d'$ in $G$, then $d$ and $d'$ would be 
isomorphic in $G_{k}$ and therefore $G_k$ would not be an ROBDD. This is illustrated 
in Figure~\ref{fig:comp_nodes} with an example.
However, in this case $d$ would not have been picked in the first place as 
this is explicitly forbidden in the Meerkat procedure. 
The inductive step proceeds along similar lines. Here, $G_k$ is assumed to be 
an ROBDD for an $i$-bit key and the high- and low-children of $d_{i+1}$ are flipped. 
\end{proof}

\begin{lemma}
\label{lem:diff}
The Boolean function $f_k$ corresponding to ROBDD $G_k$ is different from $f = f_{k^*}$ 
for all $k \neq k^*$.
\end{lemma}

\begin{proof}
$G_k$ is an ROBDD that is distinct from $G$ since at least one node in $G$ has its high and low children swapped. Because ROBDDs are canonical, the corresponding functions 
$f_k$ and $f$ must also be distinct. 
\end{proof}
  
Lemma~\ref{lem:diff} guarantees that
applying an incorrect key to the 
locked netlist results in
incorrect functionality. 


\begin{theorem}\label{thm:main}
The Meerkat locking procedure satisfies the condition for secure logic locking 
prescribed in Definition~\ref{def:sec}\add{, for the class of functions with more than $r$ nodes in their BDD that are not part of a complementary pair}.
\end{theorem}

\begin{proof}
Let $f$ and r-bit key $k^*$ 
be the input function to Meerkat. 
The ROBDD of $f$ is $G$. 
Let $D_{lock} = \{d_1, d_2, \ldots, d_r\}$ be the subset of nodes in 
$D$ 
picked by the locking procedure, and $G_{lock}$ be the corresponding locked ROBDD.
The probability of outputting $G_{lock}$ is 
$\frac{1}{{|D| \choose r}}$. Note that, for 
ease of exposition, 
we have assumed that $G$ has no complementary nodes but the proof easily 
extends to the more general setting. 

For any key $k \neq k^*$, let $G_k$ be the corresponding ROBDD 
(as per Lemma~\ref{lem:robdd}), and $f_k$ be the Boolean function that $G_k$ represents. 
As we have seen, for each node in $G$ there is a 
corresponding node in $G_k$. For all 
$i$ such that $k^*_i \neq k_i$, the high- and low-children of node 
$d_i \in D_{lock}$ of  
$G$ and $G_k$ are swapped; every other node in $G_k$ has 
the same high- and low-children as in $G$.  

Now, when $f_k$ is locked with key $k$, 
the first step is to construct its ROBDD, which is $G_k$. 
The set $D_{lock} = \{d_1, d_2, \ldots, d_r\}$ is 
picked in iterations $1$ through $r$ of the locking procedure 
with 
probability $\frac{1}{{|D| \choose r}}$ . Finally, for any $i$ for which 
the high- and low-children children of 
$d_i$ in $G_k$ and $G$ are the same (resp. swapped), the key bits $k_i$ and 
$k^{*}_i$ are also the same (resp. complementary). As a consequence, locking 
$f_k$ with key $k$ produces $G_{lock}$ with the same probability as 
locking $f$ with key $k^{*}$. 

The final step in Meerkat is to synthesize $G_{lock}$ to a locked 
netlist $C_{lock}$. 
But since locking 
$f_k$ with key $k$ produces $G_{lock}$ with the same probability as 
locking $f$ with key $k^{*}$, the same holds true for the locked netlist $C_{lock}$ 
(regardless of the synthesis approach used). 
\end{proof}


%% file: notion.tex
\subsection{Formal Notion of Security for Logic Locking}\label{sec:notion}

In this section, we describe 
a formal 
a notion of security for logic locking under the threat
scenario discussed in Section~\ref{sec:attack}. 
We re-iterate that, to the best of our knowledge, this is the \emph{first} 
time that a formal security definition has been 
proposed for logic locking. 
 



To begin, we generalize the definition of a logic locking procedure, 
$\mathscr{L}$, from Section~\ref{sec:background}. Recall that $\mathscr{L}$ 
was previously defined to take a 
synthesized netlist as input, a vulnerability that we exploited 
in our desynthesis attack. We generalize $\mathscr{L}$  
to  
take a Boolean function 
$f$ as input, along with a designer's chosen key $k^*$. 
Second, we explicitly model 
non-determinism in the locking procedure. 
We say that $Pr[\mathscr{L}(f,k^*) = C_{lock}]$ 
is the probability that $\mathscr{L}$ 
outputs locked netlist $C_{lock} \in \mathbf{C_{lock}}$ with inputs 
$f$ and $k^*$.

We define the notion of security as follows.
\begin{definition}\label{def:sec}
\note{Definition updated to address EuroS\&P'17 Review C}
A locking procedure $\mathscr{L}$ 
is secure \add{for (family of) function classes $\{\mathscr{F}_{r}\}_{r=1,2,\ldots}$} if the following condition holds for 
every function \change{$f$}{$f \in \mathscr{F}_{r}$ }, key \change{$k^*$}{$k^* \in \{0,1\}^r$} and 
locked netlist $C_{lock} \in \mathbf{C_{lock}}$ 
output by \remove{the 
locking procedure}{$\mathscr{L}$} for inputs $f$ and $k^*$:
\begin{displaymath}
P[\mathscr{L}(f,k^*) = C_{lock}] = P[\mathscr{L}(f_k,k) = C_{lock}] \\
\forall k \in \{0,1\}^r : k \neq k^{*},
\end{displaymath}
where $f_k = C_{lock}(.,k)$ is the Boolean function that $C_{lock}$ 
implements for (incorrect) key $k$. \annote{Note that by definition of a locking procedure, $f_k$ is necessarily different from $f$, for all $k \neq k^*$}{This is in response to EuroS\&P'17 Review C}.
\end{definition}

\tcignore{\ref}{1}{1}
\tcignore{\begin}{1}{1}
\tcignore{\end}{1}{0}
\tcignore{\label}{1}{0}

\begin{proposition}\label{prop:unique}
\note{New proposition added in response to EuroS\&P'17 Review C}
If $\mathscr{L}$ is a secure locking procedure, then for every function $f$, key $k^*$ and 
locked netlist $C_{lock} \in \mathbf{C_{lock}}$ 
output by $\mathscr{L}$ for inputs $f$ and $k^*$:
$$ f_{k_1} \neq f_{k_2} \, \, \forall k_1 \neq k_2,$$
where $f_k = C_{lock}(.,k)$ is the Boolean function that $C_{lock}$ 
implements for key $k$.
\end{proposition}
\begin{proof}
The proof for when $k_1 = k^*$ follows from the definition of a locking procedure. For $k_1 \neq k^*$, by definition of security, it follows that $C_{lock} \in \mathscr{L}(f_{k_1},k_1)$. $C_{lock}(.,k_2)=f_{k_2}$ must then be different from $f_{k_1}$, again, by definition of a correct locking procedure.
\end{proof}

Intuitively, Definition~\ref{def:sec} can be understood as follows. 
A 
locked netlist $C_{lock}$ implements a family 
of $2^{r}$ different Boolean functions (\add{Proposition~\ref{prop:unique}}), i.e., function $f_k$ for each key 
$k \in K$. 
Definition~\ref{def:sec} guarantees that any 
function $f_k$ in that family is as likely to have 
produced 
$C_{lock}$ (when $f_k$ is locked with key $k$) 
as the correct function $f = f_{k^*}$ locked with 
correct key $k^*$. Thus, $C_{lock}$ itself \add{provides no criteria on which to distinguish an incorrect function from the correct one, i.e., it } reveals no 
information about the \add{correct function or the } correct key.

Note that in the special case that $\mathscr{L}$ is a deterministic procedure, Definition~\ref{def:sec} 
simply says that each function $f_k$ when locked 
with key $k \in K$ should yield the same $C_{lock}$. 

The following remarks shed further light on our notion of security. 
\begin{remark}
A locking procedure that satisfies the security notion in 
Definition~\ref{def:sec} is immune to the de-synthesis attack. As discussed 
above, \emph{any} of the functions that $C_{lock}$ implements are equally likely 
to have produced $C_{lock}$, providing the attacker no criteria to distinguish between 
these functions. (Note that immunity to the desynthesis attack is guaranteed even 
for computationally unbounded attackers.)
\end{remark}

\begin{remark}
Our attack scenario 
(see Section~\ref{sec:attack}) assumes that 
an attacker has no prior information (equivalently, a prior distribution) 
on the function, $f$, that 
the designer intends 
to implement. 
The security notion in Definition~\ref{def:sec} does 
\emph{not} protect 
against a (stronger) attacker 
who has such information, specifically, 
one who
has access to a non-uniform prior distribution on the chip's 
intended functionality. 
For the sake of argument, assume that the attacker's 
prior distribution
is concentrated on a single function $f$ --- here, the attacker knows in 
advance exactly the function the designer intends to implement. 
No locking scheme is secure in this (admittedly extreme) setting. 
Provably secure logic locking 
under more general prior distributions on the 
designer's intended functionality remains an open question.

\end{remark}

\begin{remark}
Implicit in Definition~\ref{def:sec} is invariance 
to the \emph{manner} in which the function $f$ 
is described\footnote{The same function $f$ can be described in a multitude of ways, for example as a 
truth table, a netlist, or behavioral hardware description language (HDL) code.}. 
That is, an attacker cannot exploit 
prior knowledge about the way in which 
Boolean functionality is typically described 
to defeat 
logic locking procedures that are secure as per
Definition~\ref{def:sec}. In fact,
this is a vulnerability 
that our desynthesis attack exploits --- existing logic locking 
procedures operate on 
{parsimonious}
 descriptions of Boolean functionality, i.e.,  
synthesized netlist. 
In addition to addressing this vulnerability, 
our security notion buys immunity against attacks that exploit 
properties of human-generated code (coding styles, for example). 

\end{remark}

%% file: table1.tex
\begin{table*}[ht]
\centering
\caption{Performance of desynthesis attack on eight large benchmarks from the MCNC benchmark set. Numbers are average number of key bits correctly recovered by attack for 100 runs on each benchmark.}
\label{tab:desynth}
\resizebox{\textwidth}{!} {
\renewcommand{\arraystretch}{1.2}
\begin{tabular}{l|l|l|l|l|l|l|l|l|l|l|l|l|l|l|l|l|l}
\hline
\multirow{4}{*}{} & \multicolumn{16}{c}{Locking Scheme}  
		  \\ \cline{2-17}
                  & \multicolumn{8}{c|}{EPIC}                                                              & \multicolumn{8}{c}{Interference-based Locking}                                     
		  \\ \cline{2-17}
                  & \multicolumn{4}{c|}{Number of Key Bits}    & \multicolumn{4}{c|}{Number of Key Bits}    & \multicolumn{4}{c|}{Number of Key Bits}    & \multicolumn{4}{c}{Number of Key Bits}   
		  \\ \cline{2-17} 
                  & \multicolumn{4}{c|}{32}                    & \multicolumn{4}{c|}{64}      		    & \multicolumn{4}{c|}{32}                    & \multicolumn{4}{c}{64}       
		  \\ \cline{2-17} 
                  & \multicolumn{3}{l|}{Correctly Recovered} & & \multicolumn{3}{l|}{Correctly Recovered} & & \multicolumn{3}{l|}{Correctly Recovered} & & \multicolumn{3}{l|}{Correctly Recovered}&         
                  \multicolumn{1}{l}{ } \\ \cline{2-17} 
                  & Min	& Max	& Avg.	& Time (s)	& Min	& Max	& Avg.	& Time (s)	& Min	& Max	& Avg.	& Time (s)	& Min	& Max	& Avg.	& \multicolumn{1}{l}{Time (s)}
\\ \hline
dk16              & 18	& 28	& 23.42 & 599.09   	& 36    & 54    & 43.56 & 1831.11    	& 14    & 26    & 21.03 & 495.46     	& 37    &   56 	& 45.57	& \multicolumn{1}{l}{1624.08}
\\ \hline
planet1           & 17  & 27    & 20.76 & 852.18  	& 37    & 51    & 44.03 & 2932.52	& 17    & 27    & 21.82 & 771.42      	& 38    &   56  & 46.3	& \multicolumn{1}{l}{2970.52}
\\ \hline
planet            & 15  & 27    & 20.53 & 839.09	& 37    & 52    & 43.58	& 2892.06 	& 16    & 28    & 21.65 & 779.91	& 38    &   52  & 45.78 & \multicolumn{1}{l}{3204.6}
\\ \hline
s1             	  & 16  & 27    & 21.89 & 675.5		& 37    & 52    & 45.5  & 2471.49 	& 14    & 25    & 19.74 & 675.5		& 39    &   51  & 44.77 & \multicolumn{1}{l}{2815.19}
\\ \hline
sand              & 16  & 27    & 22.71 & 762.18	& 38    & 57    & 47.15 & 2762.83	& 15    & 28    & 23.46 & 794.23	& 36    &   54  & 47.1	& \multicolumn{1}{l}{3223.11}
\\ \hline
scf               & 16  & 28    & 22.24 & 1057.85	& 38    & 53    & 45.95 & 5061.97	& 14   	& 25    & 21.23 & 885.64      	& 33    &   52	& 44.1  & \multicolumn{1}{l}{3691.53}
\\ \hline
styr              & 16  & 29    & 23.51 & 685.96	& 40    & 59    & 47    & 2544.62	& 16   	& 23    & 21.11 & 612.64      	& 34    &   52	& 44.68	& \multicolumn{1}{l}{2474.07}
\\ \hline
tbk               & 14  & 27    & 21.43 & 938.56	& 33    & 50    & 42.25 & 2892.59	& 14    & 23    & 20.62 & 788.41      	& 36    &   52 	& 43.81 & \multicolumn{1}{l}{3304.9}

\\ \hline
\end{tabular}
}
\end{table*}

%% file: results.tex
\section{Experimental Evaluation and Results}\label{sec:results}

In this section, we report results that demonstrate the effectiveness of our desynthesis attack
and the area and delay overheads of Meerkat compared to 
previously proposed logic locking schemes.

Our results are presented on eight 
large benchmark netlists 
from the 
MCNC benchmark suite~\cite{LGSynth89}, which was 
specifically tailored for 
logic synthesis research. All the benchmarks we use are
finite-state machines (FSMs) (more specifically, the combinational logic part of the FSMs) taken from real-life industrial chips. 
Note that the trade secrets of a chip designer are usually in the control units of a chip's design, which are FSM-based~\cite{chakraborty2009harpoon}.

\subsection{Desynthesis Attack Results}

To test the effectiveness of the desynthesis attack 
we implemented two logic locking schemes from 
literature: EPIC~\cite{roy2008epic} and the 
interference-based locking (IBL) scheme proposed by 
Rajendran et al.~\cite{rajendran2012security}.
Both schemes use XOR/XNOR key gates but differ in the 
way locations to insert key gates are selected.
ABC~\cite{synthesis2007abc}, an academic synthesis tool 
was used to generate synthesized netlists that are then 
locked. 

Synthesized 
netlists were locked with $r=32$ and $r=64$ bit keys.
For each logic locking procedure, key size 
and benchmark combination we generate $100$ locked netlists, 
using different placements of key gates and randomly 
generated keys.

Table~\ref{tab:desynth} presents 
the results of our desynthesis attack 
on the 
eight benchmark circuits. 
We report the average, minimum and maximum number
of bits of the key that we correctly recover 
over 100 experiments. We also report the average time the attack took on each benchmark.

From the data we can make the following 
observations. (1) The success of our desynthesis 
attack is largely independent of the locking procedure 
used. This is not surprising --- neither procedure was 
designed to protect against the de-synthesis attack. 
(2) Doubling the key size roughly doubles the number of 
key bits that are correctly recovered --- 
increasing key size does not seem to reduce the effectiveness 
of our attack.
(3) Although the benchmarks vary in terms of 
characteristics (number of inputs, outputs and 
gates), the attack results are largely the same 
across benchmarks. 
(4) In the best case across 
benchmarks, locking procedures and $100$ 
runs, we 
recovered 29 of 32 keys and 59 of 64 keys. 
The probability of recovering these many bits in 
$100$ runs if indeed the locking schemes were secure (i.e., the key could at best be randomly guessed)
is overwhelmingly small.

Although not shown in Table~\ref{tab:desynth}, we 
conducted statistical hypothesis tests to 
check whether the distribution of recovered key bits 
is consistent with a random guess. For every benchmark, 
locking procedure and for both 32- and 64-bit keys 
we find a p-value of less than $0.0001$, allowing us to 
reject the null hypothesis. In other words, 
there is statistically significant evidence to believe 
that existing locking mechanisms are \emph{not} secure.

To demonstrate that the 
desynthesis attack remains effective for a 
different synthesis tool, we synthesized 
one of the benchmarks, $dk16$, 
using Cadence Encounter RTL Compiler (a commercial 
synthesis tool), and locked the synthesized netlist using EPIC and a 32-bit key. Figure~\ref{fig:cadence_atk} shows a 
histogram of the number of correctly guessed key bits over $100$ 
runs versus the histogram for the number of correctly 
guessed key bits if the attacker were randomly guessing the key.
On
average, the 
desynthesis attack recovers 
24.5 (and up to 31) 
keys, which is similar to the attack's success when 
ABC was used as a synthesis tool. Note that in 3 of the 100 runs, the attack recovered all but one key bit correctly.

\begin{figure}
    \centering
    \includegraphics[width=\columnwidth]{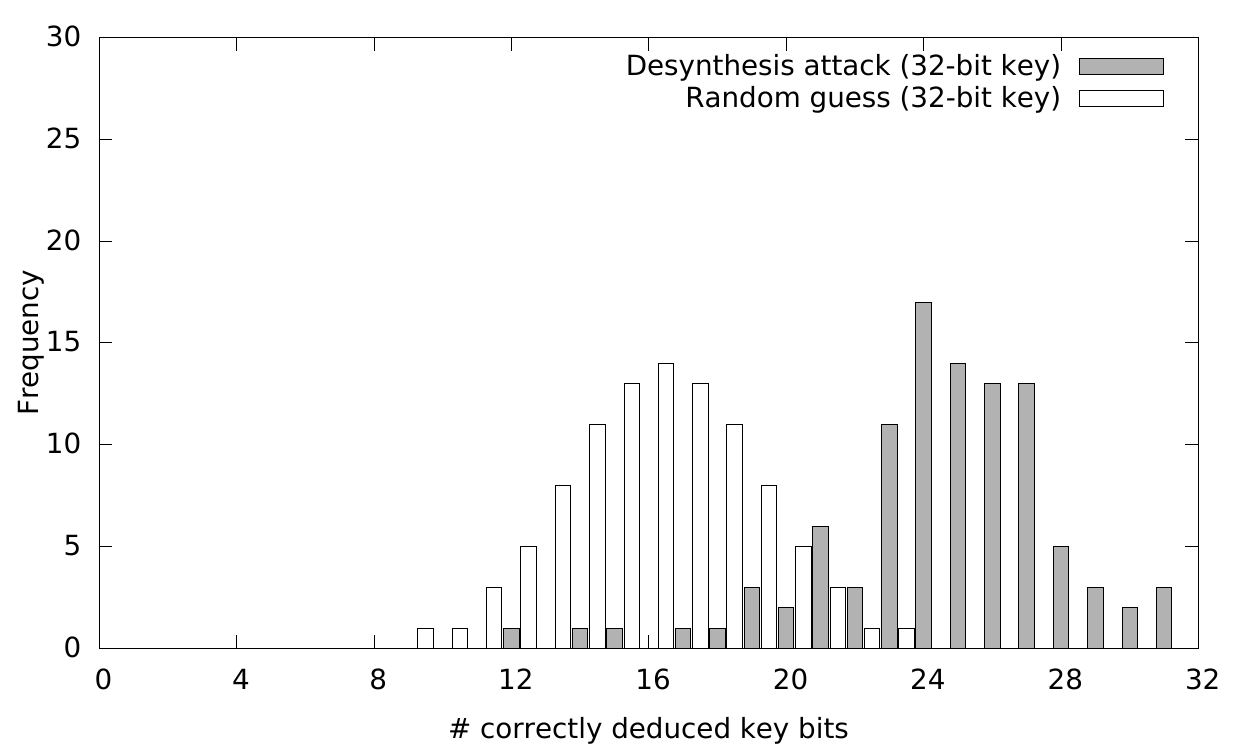}
    \caption{Performance of desynthesis attack on one benchmark from the MCNC set when the designer uses Encounter RTL Compiler to synthesize the benchmark.}
    \label{fig:cadence_atk}
\end{figure}





\subsection{Meerkat}

We have implemented Meerkat within the 
ABC synthesis environment (since ABC has an in-built 
BDD package) along with some glue code. 
In Figure~\ref{fig:area_o}, 
we compare the area overhead of logic locking 
using Meerkat with 32- and 64-bit keys over the 
original  unlocked netlist. We also show the area 
overhead of EPIC. We observe that the Meerkat overheads 
grow proportionally with key size and are modestly 
larger than EPIC overheads. The overheads are obviously 
lower for larger circuits --- Meerkat's overhead for 
64 bit key for the largest benchmark is $43\%$.  

\begin{figure}
    \centering
    \includegraphics[width=\columnwidth]{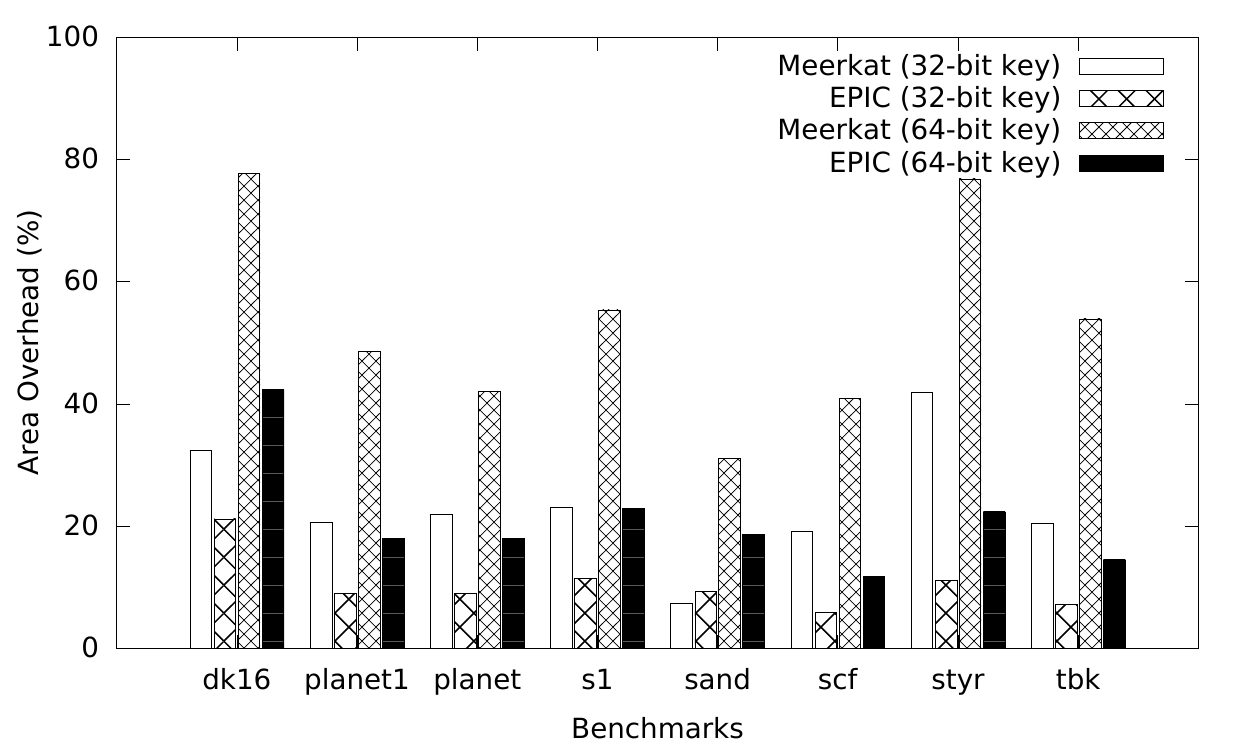}
    \caption{Area overhead for EPIC and Meerkat for locking with 32 and 64-bit keys.}
    \label{fig:area_o}
\end{figure}

Figure~\ref{fig:delay_o} plots Meerkat's delay
overheads.
Note that 
EPIC explicitly recommends 
keeping key gates off critical paths 
(paths in the netlist that impact delay) and therefore, ideally, has no delay 
overhead. On the other hand, we cannot do this in Meerkat 
because we insert key nodes \emph{before} synthesis, i.e., 
before we know where the critical paths are.

While the overheads are greater than 
those of prior work, we note that Meerkat 
is the only approach that provides formal security guarantees. 
Furthermore, as discussed in 
Section~\ref{sec:discussion}, the area and delay overheads can be amortized,  
and are significantly lower than the 
those incurred by alternative approaches to
trusted fabrication.

Finally, Figure~\ref{fig:oc} compares the 
output corruptibility (a measure of the average 
fraction 
of inputs that result in incorrect outputs 
for an incorrect key~\cite{rajendran2012logic})
of netlists locked with Meerkat 
versus those locked with EPIC. In some cases, 64-bit 
keys are required to obtain high output corruptibility 
using Meerkat.

%% file: discussion.tex
\begin{figure}
    \centering
    \includegraphics[width=\columnwidth]{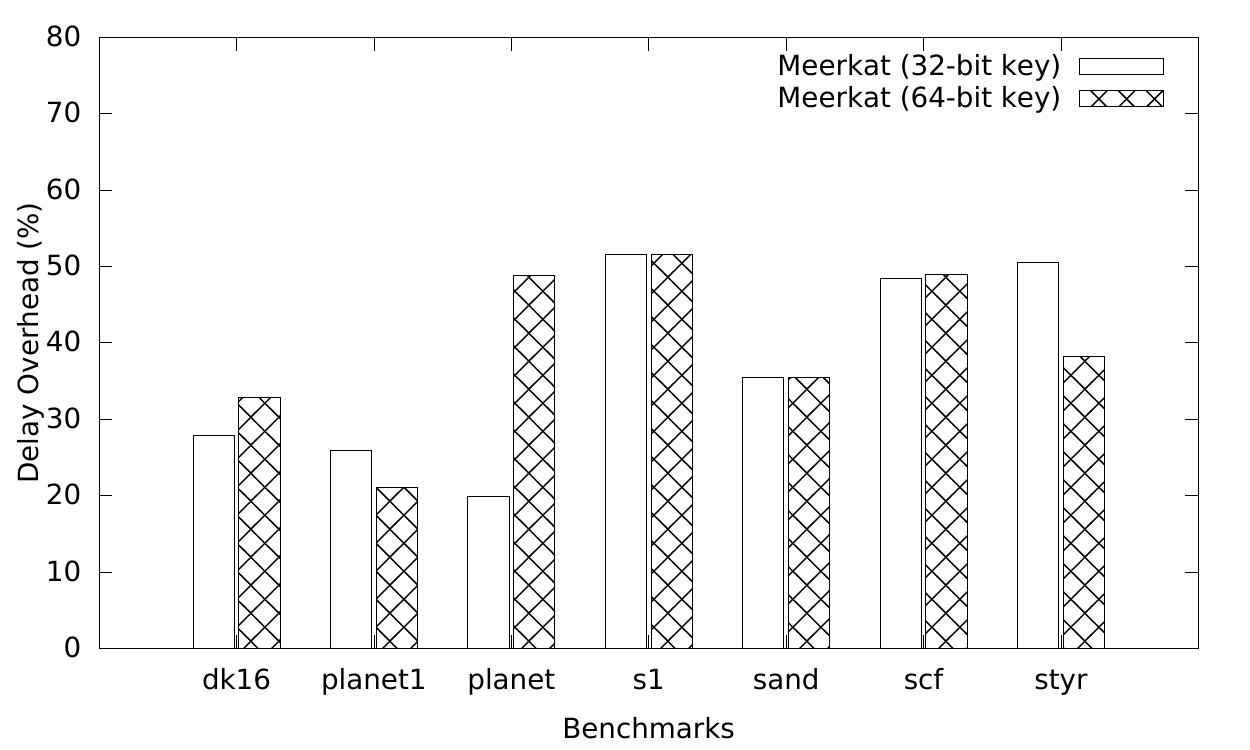}
    \caption{Delay overhead for Meerkat for locking with 32 and 64-bit keys. The tbk benchmark has no overhead for locking with 32 or 64-bit keys, and hence, is not shown.}
    \label{fig:delay_o}
\end{figure}

\begin{figure}
    \centering
    \includegraphics[width=\columnwidth]{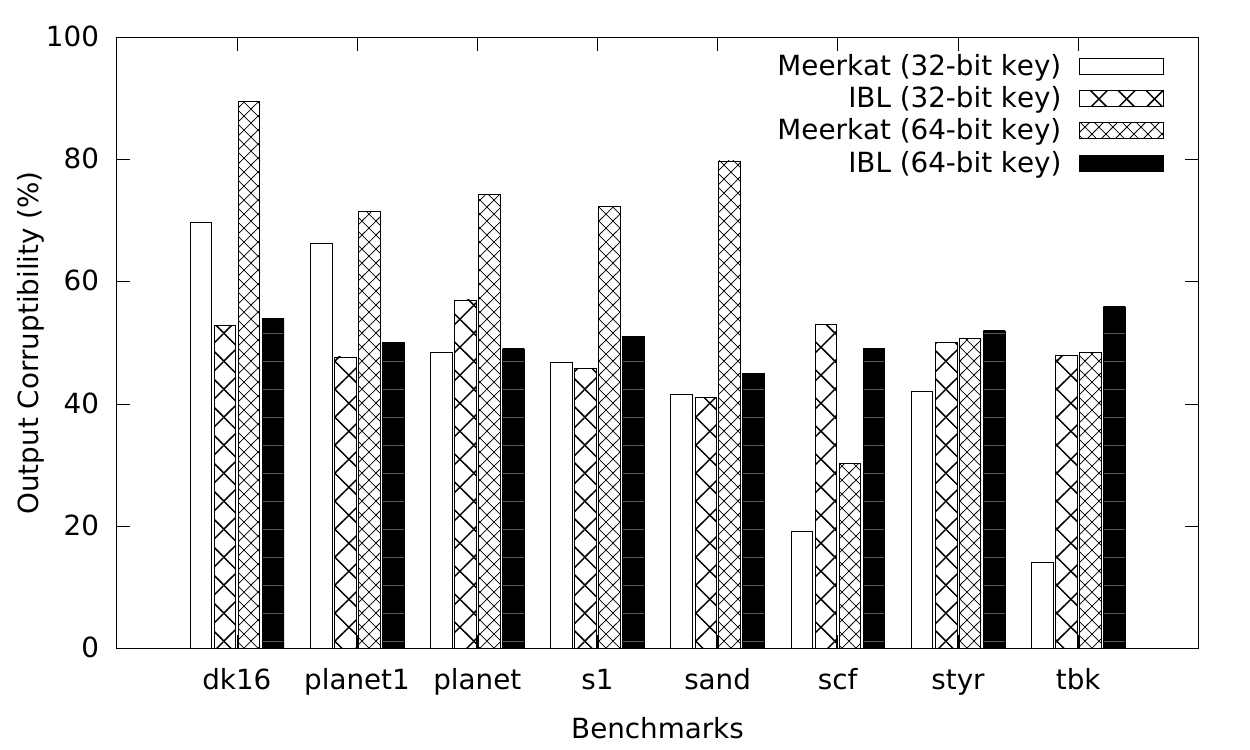}
    \caption{Output Corruptibility for EPIC and Meerkat for locking with 32 and 64-bit keys.}
    \label{fig:oc}
\end{figure}

\section{Limitations and Discussion}\label{sec:discussion}

Our desynthesis attack reveals a serious
vulnerability in existing logic locking schemes. 
Meerkat, our proposed defense mechanism, is the first 
provably secure logic locking mechanism. Certainly, our work has limitations, and in this section, we discuss these, and ways of addressing them. 
We also touch upon the distinction between the related threats of IC over-production and  
IP theft as they relate to EPIC and this work.

\paragraph{Desynthesis attack} 
In our attack model (see Section~\ref{sec:attack}), we assume that the attacker 
has access to the designer's locking procedure, including the synthesis tool. A question that arises is, is such an attacker \emph{too} strong? We first point out that this is justified based on Kerchoffs's principle~\cite{petitcolas2011kerckhoffs}, which, applied to our context, states that the system needs to be secure under the attacker we consider. Nonetheless, one may ask whether 
the EPIC procedure can be made 
more secure 
if the attacker does not know which synthesis tool 
the designer uses, or if the designer uses a
proprietary synthesis tool. 
We believe that although this merits 
further quantitative investigations (for example, 
the attacker could include the choice of synthesis tool 
in the search space of the desynthesis attack),  
we re-assert that the 
goal of our work is to provide logic locking with 
\emph{provable} guarantees on security, 
which a ``security-via-obscurity" scheme does not provide.

A second possible issue regards the observation that our desynthesis 
attack, for the cases we tried, does not appear to correctly 
identify all key bits, but only comes very close in several cases. Prior 
attacks on logic locking appear to recover all key bits, but, those attacks 
assume that the attacker can obtain inputs and outputs from a working copy of the 
chip. Furthermore, we note that our attack can be further strengthened with access to more computational resources. 
A determined 
attacker
can amplify the results by exploring a larger portion of the 
key space.
It bears mentioning that the turn-around time on 
chip fabrication can be months; 
this gives the attacker 
a significant window of time in which to 
run desynthesis attacks. \add{Further, it might be useful for an attacker to apply divide and conquer; for instance, by attacking individual logic cones within the netlist. We plan to explore this in future work.}




\paragraph{Meerkat} 
Understandably, security, particularly provable security as Meerkat provides, comes at a cost. Meerkat introduces both delay and area overheads. 
We note however that in most cases, 
the trade secrets of a designer only comprise a small part of a chip's overall design. The rest of the design consists mostly of static random-access memory (SRAM) 
(for example, a CPU's cache), 
which contains no trade secrets. 
As such, to protect a chip from piracy, a designer can choose, for locking, a small part of the overall design only, thereby significantly amortizing any overhead that is incurred. 
We locked the instruction decoder of a 32-bit RISC processor
~\cite{pmc:2016:Online}
using Meerkat and found that the total area of the processor core increased only 
by 37\%.

Still, one might wonder if a designer would be willing to incur such overhead in a highly competitive market, where producing lower-performing chips might put the designer at a disadvantage. To answer this question, we consider three options that are available to the designer: (1) The designer can forgo logic locking altogether, and instead send her IP
``in the clear" across design stages. This, of course, exposes the designer's IP to the risk of the theft.
(2) 
The designer can decide to carry out her fabrication at a trusted 
foundry (for instance, a low-end foundry located on-shore), 
instead of fabricating the 
chip at a high-end but untrusted foundry. 
However, the technology 
gap between a low-end, on-shore foundry and 
a high-end off-shore foundry can be significant. 
As an example, the most advanced 
technology available in India (an 800nm foundry)  
is 
25 years older and at least 100$\times$ 
worse (in terms of chip area alone)
than the state-of-the-art \cite{wahby2015abhi}.
(3) 
A third option is for the designer to implement her proprietary designs using reconfigurable logic. Aside from the area overheads of reconfigurable logic (around 35$\times$~\cite{fpga}), 
the designer would need
to store the configuration bits 
of the reconfigurable logic in 
read-proof memory to protect  
the chip 
from piracy.
In contrast, locking requires far fewer bits to be implemented using read-proof memory.


There are settings, however, in which it may be truly impracticable to lock a design (or even part of a design) using Meerkat.
In particular, it is known that ROBDDs for 
certain types of hardware
functions 
like multipliers and encryption 
modules are inefficiently sized~\cite{wegener1987complexity}. 
In part this is a consequence of the strong 
security guarantees we aim to 
provide ---  invariance to the function representation
implied by Definition~\ref{def:sec}
necessitates the use of a 
canonical representation.  

Yet, Meerkat inherits one of the appealing 
properties of ROBDDs that has driven 
their popularity, i.e., they are compact for a large class of Boolean functions of practical interest~\cite{brace1991efficient}. 
An interesting avenue for further research is whether netlist partitioning techniques 
can be used to mitigate 
the scalability bottleneck for functions with 
large ROBDDs, perhaps with some, known, loss in 
security. Also, we note that the 
criterion in 
Definition~\ref{def:sec} is a 
\emph{sufficient} condition; an open 
question is whether Definition~\ref{def:sec} 
can be relaxed while 
still guaranteeing security. \add{We also note that Meerkat can be applied in a LUT-based locking context. The designer could implement parts of the logic that key bits feed into with lookup tables, although this possibly results in a larger key and function space. As the original family of functions is included in the new key space, the security guarantees would remain.}

Another point of note relates to the family of 
functions within which Meerkat 
embeds the correct 
function $f$ --- we have shown that the 
$2^r-1$ alternate functions that $C_{lock}$ implements, 
all have the same ROBDD structure 
as $f$ with the high and low children of $r$ 
nodes swapped. 
Although we have shown in Section~\ref{sec:results} 
that these alternate functions 
have high output corruptability 
(output corruptability measures how different two Boolean functions are), 
it is possible to investigate broader 
families of functions
within which $f$ is embedded, for 
example functions with ROBDDs of the same size.

\paragraph{Threats: IP theft vs. over-building}
Logic locking, as 
a specific instance, has been proposed in two related 
attack scenarios: 
(1) over-building of ICs, i.e., when the foundry 
makes extra, unauthorized 
copies of 
the chip, and,
(2) IP theft, i.e., 
reverse engineering the IC's (locked) netlist. 
Defenses against over-building are referred to 
as hardware metering~\cite{koushanfar2001hardware}. 

EPIC (and succeeding logic locking schemes) claim to
protect 
against both since the protocol actually has two 
keys: a key that is common to 
all chips and is an input to key gates (this is the key we have been concerned with so far), and 
a unique key for 
every chip. 
The unique per-chip key is generated by the designer as follows: 
upon power-up, 
a newly fabricated chip generates a 
random number, using a hardware 
true random number generator (TRNG).
An on-chip 
RSA key generation circuit 
then takes this random number, interprets it as a private key, and generates the corresponding public key. 
The public key is 
read-out by the foundry 
and sent to the designer. 
The designer encrypts the 
common key using this 
public key to generate the 
chip-specific unique key (note that the unique key depends 
on the TRNG output which is different for every chip).
Finally, the unique key
is entered into the chip by the foundry, where it gets 
decrypted by an RSA decryption circuit to its plain form, the common key, that actually unlocks the netlist.


Although the unique key
protects against over-building, 
it offers no protection against 
IP theft. 
That is, our 
attack 
(and indeed all 
attacks in prior work~\cite{rajendran2012logic,subramanyan2015evaluating}) 
focus
on 
directly 
reverse engineering
the common key (recall that unlike 
prior attacks, however, we do not assume 
access to a working chip). 
In fact, a compromise of the common key 
not only reveals the designer's intended
netlist (accomplishing IP theft), but 
also enables the foundry
to generate new masks in which the 
correct key is hardwired into the circuit. 
The foundry can then mass produce functionally 
equivalent chips with this new mask, thus
effectively 
over-building the IC.
Modifying 
an existing mask or creating a new one is easily 
within the realm of capabilities of a foundry.

%% file: related.tex
\section{Related Work}

We now discuss prior art on secure outsourced 
chip fabrication, and place it in the context of our work. 

The EPIC protocol~\cite{roy2008epic} was 
the first locking scheme proposed to protect 
combinational logic from over-building 
and 
IP theft. 
The EPIC protocol is described in detail 
in Section~\ref{sec:background}. 
Succeeding work has attempted to improve upon the EPIC 
locking procedure  by
using programmable logic~\cite{baumgartenpreventing}
and MUXes~\cite{plaza2015solving, wang2016secure} 
instead 
of XOR/XNOR gates, and smart
selection of locations for 
key-gate insertion~\cite{rajendran2012logic,rajendran2015fault}.
However, \emph{all} 
these techniques assume that locking is 
performed after synthesis, 
the root of the vulnerability that 
we exploit 
in our desynthesis attack.

In terms of attacks, Maes et al.~\cite{maes2009analysis} 
demonstrated a vulnerability 
in EPIC 
using a man-in-the-middle attack 
on the communication protocol between the foundry and 
designer, thus enabling 
the attacker to over-build the IC. 
Maes et al.~\cite{maes2009analysis} also suggest a fix to 
the protocol to address this vulnerability.
Attacks on logic locking 
that aim to steal the chip's IP  
have all 
assumed that an attacker 
has at least
partial access to the chip's 
I/O functionality.  
Rajendran et al.~\cite{rajendran2012logic} 
show that 
an attacker who has a working copy of the chip
can quickly recover the key if 
the location of key gates is picked at random (as in EPIC).
They describe 
an improved technique to insert key gates 
that defends against 
this attack. 

However, Rajendran et al.'s improvement to the EPIC procedure 
was recently defeated by 
Subramanyan et al.~\cite{subramanyan2015evaluating}
using a sophisticated attack that relies 
on SAT-based inference.  
Concurrently, El Massad et al.~\cite{el2015integrated} devised
the same attack on IC camouflaging, a conceptually related hardware 
security mechanism. Liu et al.~\cite{holcomb} 
have proposed techniques to further enhance the effectiveness 
of SAT-based attacks.

Defending against these sophisticated attacks 
on logic locking still
remains an open problem. 
Recently, a defense mechanism using MUX-based 
logic locking 
was developed for a more restricted attack scenario
in which the attacker only 
has access to a certain number of chip I/O pairs, 
but not a working chip~\cite{plaza2015solving}.  
However, as we discussed in Section~\ref{sec:intro}, 
there are 
important 
settings in which the attacker has \emph{no} I/O access. 
Logic locking has so far been 
assumed to be secure in this setting. Our attack is the first attack that shows
that logic locking is insecure even when the attacker has
no I/O access.

A related line of research has been focused on
hardware metering and  
IP theft prevention for sequential logic (FSMs). 
Alkabani et al.~\cite{alkabani2007active}
were the first to propose an active metering
scheme for FSMs that modifies 
the FSM such that a 
certain initial sequence of inputs needs to be applied
to unlock the chip. 
The sequence varies from chip-to-chip and depending on 
randomness extracted from each chip. 
However, 
Alkabani et al. only aim to protect against 
over-production and not directly against IP theft. 
Further work in
this direction includes schemes that 
seek to prevent IP theft 
by locking
the (combinational)
state-transition function of the FSM~\cite{chakraborty2008hardware, chakraborty2009harpoon}. 
These techniques 
would also be vulnerable to our desynthesis attack.  
A recent paper proposes to obfuscate an FSM using structural 
transformations of the 
state-transition graph~\cite{li2013structural}, 
but these transformations cannot be used in the context 
of combinational logic locking. 

Besides IC locking, another technique that has been
proposed 
in literature for secure outsourced chip fabrication is split 
maufacturing~\cite{jvsplit,gargusenix2013,split1,vaidyanathan2014efficient}. 
The idea is to partition a chip into two (or more) 
parts and fabricate each part at a separate foundry. 
However, split manufacturing (a) requires access 
to technology like  
like 3D integration~\cite{gargusenix2013} or split fabrication~\cite{vaidyanathan2014efficient} that is still experimental, 
(b) can have high overhead~\cite{gargusenix2013} and (c) is 
susceptible to so-called proximity attacks~\cite{jvsplit}.

Other techniques to deter hardware IP theft
that are orthogonal to our work include 
watermarking~\cite{wm} (to detect IP theft) and 
IC camouflaging~\cite{camo,el2015integrated} (to defend against post-fab ``sand-and-scan" attacks). As a concluding note, besides IP theft, 
another risk of outsourced chip fabrication is an 
untrusted foundry that maliciously modifies the chip. 
We do not address this threat in our work and refer the 
reader to an excellent survey paper on this topic~\cite{tehranipoor2010survey}.

%% file: conclusion.tex
\section{Conclusion}
In this paper, we have
 proposed a new attack on logic locking 
 that, unlike prior attacks, 
 does not require an attacker to have access to a working copy of the chip. The attack exploits the observation that all 
 existing logic synthesis mechanisms lock synthesized netlists, 
 and that the synthesis step embeds information about the 
 true functionality of the chip into the locked netlist. 
 Using our attack, we were able to correctly 
 recover a significant number 
 of bits, up to 30 and 59 bits, 
 for netlists locked with 32- and 64-bit keys, respectively. 
 To address the vulnerability posed by the attack, we 
 present the first formal notion of security for 
 logic locking and develop Meerkat, a new logic locking 
 technique based on ROBDDs that is provably secure under this 
 notion. Meerkat's area and delay overheads, although greater 
 than those of existing (insecure) techniques like EPIC, are still moderate, particularly in light of the fact that they come with security guarantees.